\documentclass[a4paper,reqno]{amsart}

\calclayout

\usepackage{amsmath}
\usepackage{amssymb}
\usepackage{amsthm}
\usepackage{amsfonts}
\usepackage{mathtools}
\usepackage{nicematrix}
\usepackage{bm}
\usepackage{graphicx}
\graphicspath{ {./figures/} }
\usepackage{adjustbox}
\usepackage[inline]{enumitem}
\usepackage{url}
\usepackage{epsfig}
\usepackage{color}
\usepackage{float}
\usepackage{setspace}
\usepackage{comment}
\usepackage{appendix}
\usepackage{tikz}
\usetikzlibrary{cd}
\usetikzlibrary{decorations.pathmorphing}
\usepackage[hang,small,bf]{caption}
\usepackage[subrefformat=parens]{subcaption}
\usepackage{wrapfig}
\usepackage[bookmarksnumbered,bookmarksopen,pdfencoding=auto]{hyperref}
\hypersetup{
    hidelinks,
    linktocpage=true,
    setpagesize=false,
}
\allowdisplaybreaks

\DeclarePairedDelimiter{\abs}{\lvert}{\rvert}
\DeclarePairedDelimiter{\fat}{\lVert}{\rVert}

\theoremstyle{definition}
\newtheorem{thm}{Theorem}[section]

\newtheorem{lem}[thm]{Lemma}
\newtheorem{cor}[thm]{Corollary}
\newtheorem{prop}[thm]{Proposition}

\newtheorem{ex}[thm]{Example}

\newtheorem{defi}[thm]{Definition}
  
\theoremstyle{remark}
\newtheorem{rem}{Remark}[section]

\title{Cohomology for solutions of polygon equations}

\author{Serban Matei Mihalache}
\address{Graduate School of Mathematical Sciences, The University of Tokyo
3-8-1 Komaba, Meguro, Tokyo, 153-8914 Japan}
\email{mateimihalache@g.ecc.u-tokyo.ac.jp}

\author{Tomoro Mochida}
\address{Mathematical Institute, Tohoku University, 
6-3, Aramaki Aza-Aoba, Aoba-ku, 
Sendai, 980-8578, Japan,}
\email{tomorou.mochida.r5@dc.tohoku.ac.jp}

\begin{document}

\begin{abstract}
Polygon equations form a family of equations generalizing the pentagon equation. In this paper, we construct semi-simplicial sets of permitted colorings associated with set-theoretic solutions of polygon equations and use them to define the corresponding (co)homology groups.
We investigate several properties of these groups and establish an equivalence of categories between set-theoretic solutions of polygon equations and higher Segal semi-simplicial sets satisfying certain conditions.
As a special case, our result recovers the correspondence between bijective set-theoretic solutions of the pentagon equation and $2$-Segal semi-simplicial sets proved by Dyckerhoff--Kapranov.
\end{abstract}

\maketitle
\thispagestyle{empty}
\makeatletter

\section{Introduction}

Polygon equations were introduced by Dimakis and M\"{u}ller-Hoissen in \cite{dimakis2012kp, dimakis2015simplex} as algebraic realizations of maximal chains in higher Stasheff--Tamari orders.
They form a family of equations indexed by an integer $n \geq 3$, where the equation associated with $n$ is called the $n$-gon equation.
In the case of $n=5$, this recovers the well-known pentagon equation $T_{12}T_{13}T_{23}=T_{23}T_{12}$.
Thus, polygon equations may be thought of as higher-order analogs of the pentagon equation. 

Another family of equations closely related to polygon equations are simplex equations~\cite{bazhanov1982conditions}, which generalize the Yang--Baxter equation $R_{12}R_{13}R_{23}=R_{23}R_{13}R_{12}$~\cite{baxter1972partition,yang1967some} and the tetrahedron equation $R_{123}R_{145}R_{246}R_{356}=R_{356}R_{246}R_{145}R_{123}$~\cite{zamolodchikov1980tetrahedra,zamolodchikov1981tetrahedron}. The relationship between these equations has been studied, for example, in \cite{kashaev1998pentagon,dimakis2015simplex,dimakis2021grassmannian,kassotakis2024entwining,mihalache2025constructing}. Also,
a cohomology theory for simplex equations was defined by Korepanov, Sharygin, and Talalaev~\cite{korepanov2016cohomologies}, generalizing quandle cohomology~\cite{carter2003quandle} and Yang--Baxter cohomology~\cite{carter2004homology}.

Polygon equations are also related to higher Segal conditions introduced by Dyckerhoff and Kapranov~\cite{dyckerhoff2019higher}. The 2-Segal case was also independently studied by Gálvez-Carrillo, Kock, and Tonks~\cite{galvez2018decomposition} under the name of decomposition spaces. See also \cite{poguntke2017higher}. These conditions are higher-dimensional analogs of the ordinary Segal condition and describe when a simplex can be uniquely reconstructed from the data associated with a triangulation of a cyclic polytope. In dimension two, Dyckerhoff and Kapranov~\cite{dyckerhoff2019higher} showed that bijective set-theoretic solutions of the pentagon equation correspond to certain $2$-Segal semi-simplicial sets.

\medskip

The aim of this paper is to develop a cohomology theory for set-theoretic solutions of polygon equations and to relate the associated semi-simplicial sets to higher Segal conditions. 
The cohomology of polygon equations was considered for the $6$-gon equation by Korepanov and Sadykov~\cite{korepanov2017hexagon} (see also \cite{korepanov2019polynomial,korepanov2021nonconstant}) and for odd-gon equations (i.e., polygon equations of odd order) by Korepanov~\cite{korepanov2024odd}. Following their construction, as well as the construction in \cite{korepanov2016cohomologies} for simplex equations, this paper defines the cohomology groups of polygon equations of arbitrary order and investigates their properties. Roughly, given a solution of a polygon equation, we construct a semi-simplicial set and then define the corresponding cohomology groups:
\[
\begin{tikzcd}
    {\text{$T$: solution of $n$-gon}} &
    {\text{$\operatorname{Col}_T$: semi-simplicial set}} &
    {\text{$H^*(T)$: cohomology of $T$}}
    \arrow[rightsquigarrow, from=1-1, to=1-2]
    \arrow[rightsquigarrow, from=1-2, to=1-3]
\end{tikzcd}
\]
This semi-simplicial viewpoint also provides a natural link between the solution of polygon equations and higher Segal conditions.

We describe bases for the chain groups using colorings of certain faces. For solutions of the $n$-gon equation for small $n$, we provide an explicit boundary formula in which group cohomology (the bar construction) naturally appears. We then study several properties of these cohomology groups. In particular, we obtain some vanishing results and the relationship between the cohomology of solutions of neighboring polygon equations. Finally, we characterize the semi-simplicial sets arising from solutions of polygon equations in terms of higher Segal conditions. More precisely, we establish an equivalence between the category of solutions of polygon equations and the category of higher Segal semi-simplicial sets satisfying certain conditions in low degrees. For the pentagon equation, this recovers the correspondence of Dycherhoff and Kapranov~\cite{dyckerhoff2019higher}.

\medskip

We end this section by discussing a motivation for considering such a cohomology theory.
One motivation comes from PL topology, where polygon equations can be thought of as algebraic realizations of certain Pachner moves of triangulations (Section~\ref{sec:polygon_equations}). Solutions of polygon equations are used to construct coloring invariants of PL manifolds. A sketch of the construction is as follows:

Let 
\begin{itemize}
    \item $T$: a solution of the (dual) $n$-gon equation over a finite set $X$,
    \item $M$: an $m$-dimensional closed oriented PL manifold, such that $m\geq n-2$,
    \item $K$: a triangulation of $M$ with a fixed total order on the set of vertices.
\end{itemize}
Then we can define the set $\operatorname{Col}_T(K)$ of $T$-colorings of $K$, which consists of maps 
\[
c\colon K_{n-3}=\{\text{$(n-3)$-simplices of $K$}\}\to X
\]
such that it ``coincides'' with the map $T$ on all $(n-2)$-simplices of $K$ (Definitions~\ref{def:T-coloring} and ~\ref{def:polygon_coloring}).
We then define an integer by $|\operatorname{Col}_T(K)|$, which we want to be an invariant of $M$.
There are two obvious obstructions to this integer being an invariant, which are also pointed out in \cite{dyckerhoff2026cyclic}:
\begin{enumerate}
    \item One must show that $|\operatorname{Col}_T(K)|$ is independent of the fixed total order on the set of vertices. 
    \item In general, $T$ only guarantees the invariance of $|\operatorname{Col}_T(K)|$ under a specific Pachner move, while we need to prove invariance under all Pachner moves.
\end{enumerate}

Now let us try to generalize this construction.
One way to do this is to treat $|\operatorname{Col}_T(K)|$ as a state-sum and twist its weights.
Take a map $\omega\colon \operatorname{Col}_T(\Delta^m)\to \mathbb{C}^*$ and consider the following scalar:

\[
Z_{\omega}(K;T) \coloneqq N_K\!\!\sum_{c\in\operatorname{Col}_T(K)}\,\,\prod_{\substack{\sigma \in K_m}} \omega(c|_{\sigma})^{\varepsilon(\sigma)}, 
\]
where the sign $\varepsilon(\sigma)$ is $1$ if the orientation induced by the vertex order agrees with the orientation from $M$, and $-1$ otherwise, and $N_K$ is an appropriate normalization term. When $\omega\equiv 1$, this recovers the original scalar $|\operatorname{Col}_T(K)|$. Again, for this to be invariant under Pachner moves, $\omega$ has to satisfy some cocycle conditions. 
This is what motivates us to define and study the cohomology associated with solutions of the $n$-gon equation.

In the $n=4$ case, a solution of the dual $4$-gon equation is just an associative product on $X$.
If we further assume that this product forms a group, the resulting coloring invariant turns out to be the untwisted Dijkgraaf--Witten invariant of $M$. Furthermore, if we consider the cohomology of the dual 4-gon equation (this cohomology turns out to match the group cohomology), the resulting cocycle invariant coincides with the Dijkgraaf--Witten invariant~\cite{dijkgraaf1990topological,wakui1992dijkgraaf}.

For readers familiar with the theory of (bi)quandles, we note that this is closely parallel to the theory of (bi)quandle colorings and (bi)quandle cocycle invariants of knots or surface knots~\cite{carter2003quandle}. In that setting, solutions of the Yang--Baxter equation and the tetrahedron equation play a role in ensuring the invariance under certain moves of diagrams of knots or surface knots.
\medskip

The rest of the paper is organized as follows: In Section~\ref{sec:polygoncohomology}, we define a semi-simplicial set associated with an arbitrary subset and then define a chain complex and its (co)homology. In Section~\ref{sec:polygon_equations}, we define polygon equations and introduce certain subsets associated with their solutions. We then define the cohomology of a solution of a polygon equation to be the cohomology of the corresponding subset. In Section~\ref{sec:chain_structure}, we study the structure of these chain complexes and give bases for their chain groups. Section~\ref{sec:examples} presents some examples. In particular, we consider the cohomology of a solution of the $n$-gon equation for small $n$. In Section~\ref{sec:properties}, we study properties of polygon cohomology, some of which are related to our previous work~\cite{mihalache2025constructing}. 
In Section~\ref{sec:segal}, we establish a correspondence between solutions of polygon equations and higher Segal semi-simplicial sets.

\subsection*{Acknowledgment}

We would like to thank Yuji Terashima for helpful discussions. The second author is supported by JSPS KAKENHI Grant Number JP26KJ0522.

\section{Definition of Polygon (co)homology}\label{sec:polygoncohomology}

\subsection{Notation}

Let us fix the notation used throughout the paper.
Let $\Delta^{N}$ be the standard $N$-simplex, viewed as the simplicial complex whose vertex set is an ordered set $\{0,1, \ldots , N\}$ and whose simplices are all subsets of this vertex set:
\begin{equation*}
    \Delta^N=[0,1,\ldots, N].
\end{equation*}
Since the vertices of $\Delta^{N}$ are ordered, we identify an $n$-simplex $\sigma\subset \Delta^{N}$ with a set of vertices 
\begin{equation*}
    \sigma = [v_0, v_1, \ldots , v_n] \quad (v_0 < v_1 < \cdots < v_n).
\end{equation*}
For an $n$-simplex $\sigma=[v_0, v_1, \ldots , v_n]\subset \Delta^{N}$ and $0 \leq i \leq n$, set
\begin{equation*}
    \partial_i \sigma \coloneqq [v_0,\ldots,\hat{v_i},\ldots,v_n],
\end{equation*}
which is an $(n-1)$-simplex obtained by removing the $i$-th vertex $v_i$ of $\sigma$.
We denote by $\Delta^{N}_{n}$ the set of all $n$-simplices in $\Delta^{N}$.
\medskip

\subsection{Definition of (co)homology}\label{subsec:poly_cohomology}

We fix an integer $n \geq 3$.
Let $N\geq n-3$ be an integer, $X$ be a set, and 
\begin{equation}\label{eq:subset}
    T \subset X^{\times (n-1)}.
\end{equation}

\begin{defi}\label{def:T-coloring}
    A $\boldsymbol{T}$\textbf{-coloring} of $\Delta^{N}$ is a map $c:\Delta^{N}_{n-3} \to X$
such that for each $\sigma \in \Delta^N_{n-2}$, 
\begin{align*}
   (c(\partial_0 \sigma), c(\partial_1 \sigma), \ldots , c(\partial_{n-2} \sigma)) \in T.
\end{align*}
When $N=n-3$, we define a $T$-coloring to be any map $c\colon \Delta^{n-3}_{n-3}=\{[0,\ldots,n-3 ]\} \to X$. 
We denote by $\operatorname{Col}_T(\Delta^N)$ the set of all $T$-colorings of $\Delta^N$.
\end{defi}

\begin{rem}
    After defining the (dual) $n$-gon equation and its solution $T$ in Section \ref{sec:polygon_equations}, we will only consider subsets $T\subset X^{\times (n-1)}$ of the form given in Definition \ref{def:polygon_coloring}.
\end{rem}

\subsection*{Semi-simplicial set}
Given a subset $T$ of $X^{\times (n-1)}$, we define a semi-simplicial set $\operatorname{Col}_T$ by
\begin{equation*}
    (\operatorname{Col}_T)_N = \left\{\begin{array}{cc}
       \operatorname{Col}_T(\Delta^N) &  \text{if $N\geq n-3$,}\\
       \{*\}  & \text{if $N< n-3$}
    \end{array}\right.
\end{equation*}
and the face maps $d_i^N \colon (\operatorname{Col}_T)_N \to (\operatorname{Col}_T)_{N-1}$ are given as follows:
Let $N \geq n-2$.
For $0\leq i \leq N$, define $\varepsilon_i\colon \{0,\ldots,N-1\}\to \{0,\ldots,N\}$ by
\begin{align*}
\varepsilon_i(v)\coloneq
    \begin{cases}
        v & v < i,\\
        v+1 & v \geq i.
    \end{cases}
\end{align*}
For $\sigma = [v_0,\ldots,v_{n-3}]\in\Delta^{N-1}_{n-3}$, set 
\begin{align*}
    \delta^i_N(\sigma)\coloneq [\varepsilon_i(v_0,\ldots,\varepsilon_i(v_{n-3})].
\end{align*}
Then, for $0 \leq i \leq N$ and a $T$-coloring $c\in \operatorname{Col}_T(\Delta^N)$, we define the $T$-coloring $d_i^N(c) \in \operatorname{Col}_T(\Delta^{N-1})$ by 
\begin{align*}
    (d_i^N(c))(\sigma) \coloneqq c(\delta_N^i(\sigma)) \qquad \text{for } \sigma\in\Delta^{N-1}_{n-3}.
\end{align*}
In other words, we view $\Delta^{N-1}$ as a subcomplex of $\Delta^N$ by canonically identifying $\Delta^{N-1}$ as $\partial_i \Delta^N$ in a manner that preserves the order of the vertices. Then, $d_i^N(c)$ is defined by the restriction 
\begin{equation*}
    c\rvert_{(\partial_i\Delta^{N})_{n-3}} \colon \Delta^{N-1}_{n-3} \cong (\partial_i\Delta^N)_{n-3} \subset \Delta^N \to X.
\end{equation*} 
It is then easy to check that $d_id_j=d_{j-1}d_i$ for all $i<j$.

\subsection*{Chain complex and (co)homology}
For a set $X$, $\mathbb{Z}\langle X \rangle$ denotes the free abelian group generated by $X$.

We define the abelian group $C_N(T)$ by 
\begin{equation*}
    C_N(T) \coloneqq \left\{\begin{array}{cc}
       \mathbb{Z}\langle\operatorname{Col}_T(\Delta^N)\rangle  &  \text{if $N\geq n-3$,}\\
       0  & \text{if $N< n-3$,}
    \end{array}\right.
\end{equation*}
and the abelian group $C^N(T)$ is defined as the dual of $C_N(T)$:
\begin{equation*}
    C^N(T) \coloneqq \left\{\begin{array}{cc}
        \operatorname{Map}(\operatorname{Col}_T(\Delta^N)\,,\, \mathbb{Z}) & \text{if $N \geq n-3$,} \\
        0 & \text{if $N < n-3$.}
    \end{array}\right.
\end{equation*}
The (co)boundary maps $\partial_N \colon C_N(T) \to C_{N-1}(T)$ and $\delta^N \colon C^{N-1}(T) \to C^{N}(T)$ are, as usual, defined by 
\begin{equation*}
    \partial_N \coloneqq \sum_{i=0}^{N} (-1)^{i} d_i^N\quad\text{and}\quad \delta^N \coloneqq \partial_{N}^*
\end{equation*}
for $N\geq n-2$, and $\partial_N\coloneqq0$ and $\delta^N\coloneqq0$ for $N < n-2$.

\begin{defi}
     Let $T\subset X^{\times (n-1)}$ be a subset. For an abelian group $A$, we define the chain complex $(C_\bullet(T;A), \partial_\bullet)$ and the cochain complex $(C^{\bullet}(T;A), \delta^\bullet)$ by
    \begin{alignat*}{2}
        C_\bullet(T;A) &\coloneqq C_\bullet(T) \otimes A, \quad & \partial_\bullet &\coloneqq \partial_\bullet \otimes \mathrm{id}_A, \\
        C^\bullet(T;A) &\coloneqq \operatorname{Hom}(C_\bullet(T), A), \quad & \delta^\bullet &\coloneqq \operatorname{Hom}(\partial_\bullet,\mathrm{id}_A)
    \end{alignat*}
    in the usual way. The \textbf{polygon homology} and \textbf{cohomology} groups with coefficients in $A$ are defined by 
    \begin{equation*}
        H_*(T;A) \coloneqq H_*(C_{\bullet}(T;A), \partial_{\bullet}), \quad H^*(T;A) \coloneqq H^*(C^{\bullet}(T;A), \delta^{\bullet}),
    \end{equation*}
    respectively. As usual, if the coefficients are in $\mathbb{Z}$, we write $H_*(T)$ and $H^*(T)$ for $H_*(T;\mathbb{Z})$ and $H^*(T;\mathbb{Z})$, respectively.
\end{defi}

\begin{rem}
    Let $\fat{\operatorname{Col}_T}$ be the (fat) geometric realization, and $\fat{\operatorname{Col}_T}^{(n-4)}$ be the $(n-4)$-skeleton of $\fat{\operatorname{Col}_T}$.
    Then, it is easy to see that
    \begin{equation*}
        H_*(T;A) \cong H_*(\fat{\operatorname{Col}_T}, \fat{\operatorname{Col}_T}^{(n-4)};A),
    \end{equation*}
    i.e., the relative homology of the pair $(\fat{\operatorname{Col}_T}, \fat{\operatorname{Col}_T}^{(n-4)})$.
\end{rem}

\section{Polygon equations}\label{sec:polygon_equations}

Originally, polygon equations were introduced in \cite{dimakis2012kp, dimakis2015simplex} as algebraic realizations of higher Stasheff--Tamari orders.
In this section, following \cite{kashaev2015realizations}, we define the $n$-gon equation in terms of an algebraic realization of a certain Pachner move on the standard simplex $\Delta^{n-1}$.
For the diagrammatic definition, see \cite{mihalache2025constructing}.

\subsection{Definition of polygon equation}

Let $n\geq 3$ be an integer, and $\Delta^{n-1}=[0,1,\ldots, n-1]$ be the standard $(n-1)$-simplex. 
Given a set $X$ and a map
\[
T\colon X^{\times\lfloor\frac{n-1}{2}\rfloor}\to X^{\times\lfloor\frac{n}{2}\rfloor},
\]
we define the \textbf{set-theoretic polygon equation} as follows:
\begin{itemize}
    \item To each $(n-3)$-simplex $s \in \Delta^{n-1}$, we associate a copy of the set $X$, denoted by $X(s)$.
    \item To each $(n-2)$-simplex $p \in \Delta^{n-1}$, we associate a copy of the map $T$, denoted by $T(p)$: 
\begin{equation*}
    T(p)\colon X(\partial_1 p)\times X(\partial_3 p)\times\cdots\times X(\partial_{2\lfloor\frac{n-1}{2}\rfloor-1} p)\to X(\partial_0 p)\times X(\partial_2 p)\times\cdots\times X(\partial_{2\lfloor\frac{n}{2}\rfloor -2} p).
\end{equation*}
\end{itemize}
As indicated, each input factor and output factor of $T(p)$ are labeled by the $(n-3)$-simplex in the boundary of $p$.

Next we partition the set of $(n-2)$-dimensional faces of $\Delta^{n-1}$ into $D^+$ and $D^-$:
\begin{equation*}
    D^{+} \coloneqq \{\,\partial_i\Delta^{2k}\,|\,\text{$0\leq i \leq 2k$, $i$ is even}\,\},\qquad 
    D^{-} \coloneqq \{\,\partial_j\Delta^{2k}\,|\,\text{$0\leq j \leq 2k$, $j$ is odd}\,\}
\end{equation*}
We compose the maps $T(\partial_i\Delta^{n-1})$ for $\partial_i\Delta^{n-1} \in D^{+}$ along the matching labeled outputs and inputs, and arrange the remaining free outputs and free inputs, each in reverse lexicographic order. We do the same for $T(\partial_j\Delta^{n-1})$ for $\partial_j\Delta^{n-1} \in D^{-}$. The resulting equality between these two compositions is called the \textbf{$\boldsymbol{n}$-gon equation}. Figure~\ref{fig:graphical45goneq} is a diagrammatic picture of polygon equations. 

Similarly, we define the \textbf{set-theoretic dual polygon equation} as follows.
For a map $S\colon X^{\times\lfloor\frac{n}{2}\rfloor}\to X^{\times\lfloor\frac{n-1}{2}\rfloor}$, we associate 
\begin{equation*}
    S(p)=S\colon X(\partial_0 p)\times X(\partial_2 p)\times\cdots\times X(\partial_{2\lfloor \frac{n}{2}\rfloor -2} p)\to X(\partial_1 p)\times X(\partial_3 p)\times\cdots\times X(\partial_{2\lfloor\frac{n-1}{2}\rfloor-1} p)
\end{equation*}
to each $(n-2)$-simplex $p\subset\Delta^{n-1}$.
The exact same procedure as above yields another equation called the \textbf{dual $\boldsymbol{n}$-gon equation}.

\begin{figure}[h]
    \centering
    \begin{subcaptionblock}{.4\textwidth}
        \centering
        \includegraphics[width=.9\linewidth]{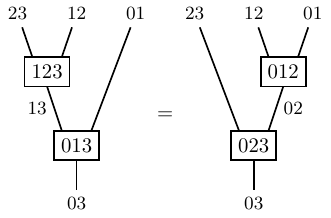}
        \caption*{4-gon equation associated to $\Delta^3$}
    \end{subcaptionblock}
    \begin{subcaptionblock}{.4\textwidth}
        \centering
        \includegraphics[width=.9\linewidth]{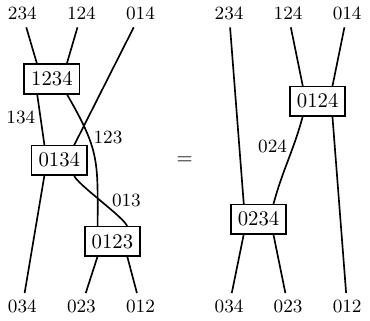}
        \caption*{5-gon (pentagon) equation associated to $\Delta^4$}
    \end{subcaptionblock}
    \caption{Diagrammatic pictures of 4- and 5-gon equations.
    The box labeled with a simplex $p$ represents the map $T(p)$, and each edge labeled with a simplex $s$ represents the set $X(s)$. 
    Read from bottom to top. Reading from top to bottom yields the dual 4- and dual 5-gon equations.}
    \label{fig:graphical45goneq}
\end{figure}

\begin{ex}
    Polygon equations for lower orders are as follows:
    \begin{itemize}
        \item A solution of the (dual) $3$-gon equation is a projector, 
        \[
        P\colon X\to X \quad \text{such that} \quad PP=P.
        \]
        \item A solution of the $4$-gon equation is a coassociative coproduct,
        \[
        \Delta\colon X\to X\times X \quad \text{such that} \quad (\Delta\times \mathrm{id}_X)\Delta=(\mathrm{id}_X\times\Delta)\Delta,
        \]
        and a solution of the dual $4$-gon equation is an associative product,
        \[
        M\colon X\times X\to X \quad \text{such that} \quad M(M\times \mathrm{id}_X)=M(\mathrm{id}_X\times M).
        \]
        \item A solution of the 5-gon equation is
        \[
        T\colon X\times X\to X\times X \quad \text{such that} \quad T_{12}T_{13}T_{23}=T_{23}T_{12},
        \]
         where $T_{ij}$ acts on the $i$-th and $j$-th factors of $X^{\times 3}$ as $T$ and on the rest as the identity. A solution of the dual $5$-gon equation is
         \[
         S\colon X\times X\to X\times X \quad \text{such that} \quad S_{23}S_{13}S_{12}=S_{12}S_{23}.
         \]
         
         Many examples of set-theoretic solutions of the pentagon equation are known. See, for example, \cite{mazzotta2025set} and the references therein.
         Furthermore, when $X$ is a finite set, bijective solutions over $X$ are completely classified in \cite{colazzo2024bijective}. See also \cite{colazzo2020set, castelli2026commutative}. 
    \end{itemize}

    For the general structure of the $n$-gon equation for higher $n$, see, for example, \cite{muller2024structure}. There are also known families of non-constant\footnote{Here, a non-constant solution means a solution in which the maps $T$ appearing in the polygon equation are allowed to be different.} solutions of polygon equations~\cite{dimakis2021grassmannian, wan2024matrix}.
    In general, examples of (set-theoretic) solutions of the (dual) $n$-gon equations are still rare for $n\geq 6$.
    
\end{ex}

\begin{ex}\label{ex:sol_commmonoid}
    In \cite{mihalache2025constructing}, we showed that for a commutative monoid $X$, the map $T\colon X^{\times\lfloor\frac{n-1}{2}\rfloor}\to X^{\times\lfloor\frac{n}{2}\rfloor}$ defined by
\begin{equation}\label{eq:Tcommutativemonoid}
    \left\{\;
    \begin{aligned}
        T(x_1,\ldots,x_k) &\coloneqq (x_1,\,x_1x_2,\,x_2x_3,\,\ldots,\,x_{k-2}x_{k-1},\,x_{k-1}x_k) && \text{if $n=2k+1$,}\\
        T(x_1,\ldots,x_{k-1}) &\coloneqq (x_1,\,x_1x_2,\,x_2x_3,\,\ldots,\,x_{k-2}x_{k-1},\,x_{k-1}) && \text{if $n=2k$}
    \end{aligned}  
    \right.
\end{equation}
is a solution of the $n$-gon equation,\footnote{For $n\leq6$, commutativity is not necessary.} and the map $S\colon X^{\times\lfloor\frac{n}{2}\rfloor}\to X^{\times\lfloor\frac{n-1}{2}\rfloor}$ defined by
\begin{equation}\label{eq:Scommutativemonoid}
    \left\{\;
    \begin{aligned}
        S(x_1,\ldots,x_k) &\coloneqq (x_1x_2,\,x_2x_3,\,\ldots,\,x_{k-2}x_{k-1},\,x_{k-1}x_k,\,x_k) && \text{if $n=2k+1$,}\\
        S(x_1,\ldots,x_{k}) &\coloneqq (x_1x_2,\,x_2x_3,\,\ldots,\,x_{k-2}x_{k-1},\,x_{k-1}x_k) && \text{if $n=2k$}
    \end{aligned}  
    \right.
\end{equation}
is a solution of the dual $n$-gon equation.\footnote{For $n\leq5$, commutativity is not necessary.}
\end{ex}

For more detailed treatments of solutions of polygon equations, see \cite{muller2024structure}.

In what follows, we are interested in the case where the subset $T$ in \eqref{eq:subset} is obtained from a solution of a (dual) polygon equation.

\begin{defi}\label{def:polygon_coloring}
    For set-theoretic solutions $T$ and $S$ of the $n$-gon and dual $n$-gon equations, respectively, we define subsets \eqref{eq:subset} of $X^{\times (n-1)}$ as
    \begin{align*}
        T &= \{\, (x_0, x_1, \ldots ,x_{n-2}) \subset X^{\times (n-1)} \mid T(x_1, x_3, \ldots x_{2\lfloor\frac{n-1}{2}\rfloor-1}) = (x_0, x_2, \ldots ,x_{2\lfloor\frac{n}{2}\rfloor-2}) \,\}, \\
        S &= \{\, (x_0, x_1, \ldots ,x_{n-2}) \subset X^{\times (n-1)} \mid S(x_0, x_2, \ldots ,x_{2\lfloor\frac{n}{2}\rfloor-2})=(x_1, x_3, \ldots x_{2\lfloor\frac{n-1}{2}\rfloor-1}) \,\}.
    \end{align*}
    (By abuse of notation, we use the same symbol to denote both a solution of a (dual) polygon equation and the corresponding subset.) Accordingly, the associated semi-simplicial set, chain complex, and cohomology groups are denoted by $\operatorname{Col}_T$, $C_{\bullet}(T;A)$, and $H^*(T;A)$, respectively, as in Section~\ref{sec:polygoncohomology}.
\end{defi}

\begin{ex}
Let $M\colon X \times X \to X$ be a solution of the dual 4-gon equation, i.e., an associative product.
An $M$-coloring of $\Delta^2$ or $\Delta^3$ is a coloring of its 1-simplices by elements of $X$ satisfying the following conditions:
    \begin{figure}[H]
    \centering
    \includegraphics[]{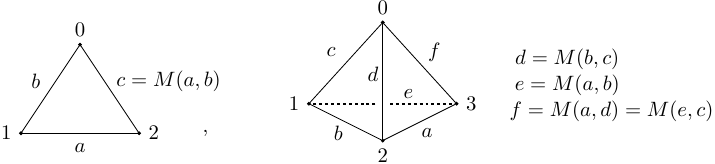}
\end{figure}
\end{ex}

\section{Structure of chain complexes}\label{sec:chain_structure}

For the rest of the paper, we assume that $T$ is a set-theoretic solution of the $n$-gon equation and $S$ is a set-theoretic solution of the dual $n$-gon equation.

\subsection{Structure of $T$-colorings}

For an $(m+1)$-simplex $\sigma \subset \Delta^N$ and an $m$-simplex $\tau \subset \sigma$, we say that
\begin{itemize}
    \item $\tau$ is positive with respect to $\sigma$ (or $\tau$ is a positive face of $\sigma$) if $\partial_i \sigma = \tau$ for even $i$,
    \item $\tau$ is negative with respect to $\sigma$ (or $\tau$ is a negative face of $\sigma$) if $\partial_i \sigma = \tau$ for odd $i$.
\end{itemize}

\begin{defi}
    An $m$-simplex $\tau\subset\Delta^N$ is called \textbf{absolutely positive} (resp. \textbf{absolutely negative}) if, for every $(m+1)$-simplex $\sigma\subset\Delta^N$ with $\tau\subset\sigma$, $\tau$ is positive (resp. negative) with respect to $\sigma$ (Figure~\ref{fig:abs_edges}).
\end{defi}

For integers $N>m$, we denote by $\mathcal{P}^N_m$ and $\mathcal{N}^N_m$ the sets of all absolutely positive and absolutely negative $m$-simplices of $\Delta^N$, respectively.

\begin{figure}[h]
    \centering
    \includegraphics[]{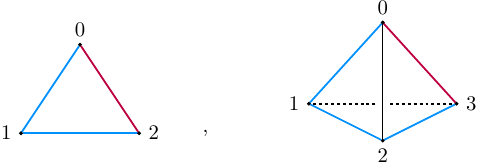}
    \caption{Example of absolutely \textcolor{blue}{positive}/\textcolor{red}{negative} 1-simplices of $\Delta^2$ and $\Delta^3$.}
    \label{fig:abs_edges}
\end{figure}

\begin{prop}\label{prop:absolutely_positeve/negative}
The sets of absolutely positive/negative simplices are given as follows:
    \begin{enumerate}[label=(\arabic*)]
        \item $\mathcal{P}^N_{2k-1}=\\ \{\, [v_0, \ldots, v_{2k-1}] \subset \Delta^N \mid 0\leq v_0<\cdots<v_{2k-1} \leq N \text{ and } v_{2j+1}=v_{2j}+1 \text{ for } 0 \leq j \leq k-1 \,\}$.\label{prop:absolutely_positeve/negative1}
        \item $\mathcal{N}^N_{2k-1}=\\ \{\, [0,v_1, \ldots, v_{2k-2},N] \subset \Delta^N\mid 0<v_1<\cdots<v_{2k-2}<N \text{ and } v_{2j}=v_{2j-1}+1 \text{ for } 1\leq j \leq k-1 \,\}$.\label{prop:absolutely_positeve/negative2}
        \item $\mathcal{P}^N_{2k}=\\ \{\, [v_0,\ldots,v_{2k-1},N] \subset \Delta^N \mid 0\leq v_0<\cdots<v_{2k-1}< N \text{ and } v_{2j+1}=v_{2j}+1 \text{ for } 0 \leq j \leq k-1 \,\}$.\label{prop:absolutely_positeve/negative3}
        \item $\mathcal{N}^N_{2k}=\\ \{\, [0,v_1,\ldots,v_{2k}] \subset \Delta^N \mid 0 < v_1<\cdots<v_{2k} \leq N \text{ and } v_{2j}=v_{2j-1}+1 \text{ for } 1 \leq j \leq k \,\}$.\label{prop:absolutely_positeve/negative4}
    \end{enumerate}
\end{prop}

\begin{proof}
    
    For \ref{prop:absolutely_positeve/negative1}, we first show that the elements of $\mathcal{P}^N_{2k-1}$ are absolutely positive.
    Consider $\tau=[v_0, \ldots, v_{2k-1}]\in \mathcal{P}^N_{2k-1}$. Any $2k$-simplex $\sigma \subset \Delta^N$ with $\tau \subset \sigma$ is given by appending a vertex $w\in\{0,1,\ldots ,N \} \backslash \{v_0, \ldots, v_{2k-1}\}$ to $\tau$. 
    Since \begin{enumerate*}[label=(\Roman*)]
        \item $v_{2j+1}=v_{2j}+1$ for $0 \leq j \leq k-1$, i.e., there is no $w$ such that $v_{2j}<w<v_{2j+1}$, and
        \item $2k-1$ is odd,
    \end{enumerate*}
     the vertex $w$ appended to $\tau$ will always lie at an even position $i$ with $\partial_i \sigma=\tau$.
    Thus, $\tau$ is absolutely positive. Next, consider a $(2k-1)$-simplex $\tau=[v_0, \ldots, v_{2k-1}] \subset \Delta^N$ such that $\tau \notin \mathcal{P}^N_{2k-1}$. Then there exists $w\in \{0,1,\ldots ,N \}$ such that $v_{2j}<w<v_{2j+1}$ for some $0 \leq j \leq k-1$.
    If we append such $w$ to $\tau$ and form the $2k$-simplex $\sigma\subset\Delta^N$, then $w$ lies at an odd position $i$ and $\partial_i\sigma=\tau$. Thus, $\tau$ is not absolutely positive.
    
    The other cases can be shown similarly.
\end{proof}

\begin{rem}
    Note that the sets $\mathcal{P}^N_m$ and $\mathcal{N}^N_m$ correspond to the lower or upper triangulations of the cyclic polytope $C([N],m)$.
    Thus the above proposition is essentially Gale's evenness criterion~\cite{gale1963neighborly}.
\end{rem}

\begin{lem}\label{lem:rank_of_chain}
    We have
    \begin{equation*}
        \abs*{\mathcal{P}^N_{2k-1}} = \binom{N+1-k}{k},\quad \abs*{\mathcal{N}^N_{2k-1}} = \binom{N-k}{k-1},\quad \abs*{\mathcal{P}^N_{2k}} = \binom{N-k}{k},\quad \abs*{\mathcal{N}^N_{2k}} = \binom{N-k}{k}.
    \end{equation*}
\end{lem}

\begin{proof}
    We prove the first one. The others can be proved similarly.
    
    We obtain recursively the relation
    \begin{equation*}
        \abs*{\mathcal{P}^N_{2k-1}} = \abs*{\mathcal{P}^{N-1}_{2k-1}} + \abs*{\mathcal{P}^{N-2}_{2k-3}}.
    \end{equation*}
    This reminds us of Pascal's identity
    \begin{equation*}
        \binom{n}{r} = \binom{n-1}{r} + \binom{n-1}{r-1}
    \end{equation*}
    for binomial coefficients. Substituting $r=k$ and $n=N-k+1$ yields the claim.
\end{proof}

We make some preparations for the proof of Proposition~\ref{prop:coloring_of_absolutely_positive/negative} below.
    Let us define a strict partial order on $\Delta^N_{n-3}$.
    For simplices $\tau, \tau^{\prime}\in\Delta^N_{n-3}$, define $\tau \lessdot \tau^{\prime}$ if there exists $\sigma\in\Delta^N_{n-2}$ such that $\tau$ and $\tau^{\prime}$ are negative and positive with respect to $\sigma$, respectively.
    Set $\tau < \tau^{\prime}$ if there is a sequence $\tau \lessdot \tau_1 \lessdot \cdots \lessdot \tau_l \lessdot \tau^{\prime}$. The next lemma shows that this is a strict partial order.
    
\begin{lem}
    $<$ is a strict partial order on $\Delta^N_{n-3}$.
\end{lem}

\begin{proof}
    We relate $<$ to the alternating lexicographical order $<_{\mathrm{alt}}$. Recall that for $\tau=[v_0,\ldots,v_{n-3}], \tau'=[v'_0,\ldots,v'_{n-3}]\in\Delta^N_{n-3}$ and the first index $r$ for which $v_r\ne v'_r$, $\tau<_{\mathrm{alt}}\tau'$  if either
    \begin{equation*}
        \text{$r$ is even and $v_r<v'_r$, or $r$ is odd and $v_r>v'_r$.}
    \end{equation*}
    
    Now assume that $\tau\lessdot\tau'$. Then there is $\sigma=[s_0,\ldots,s_{n-2}]\in\Delta^N_{n-2}$ such that $\tau=\partial_i\sigma$ for odd $i$ and $\tau'=\partial_j\sigma$ for even $j$. We show that $\tau<_{\mathrm{alt}}\tau'$.
    
    If $i<j$, then $i$ is the first place in which $\tau=\partial_i\sigma$ and $\tau'=\partial_j\sigma$ differ.
    The $i$-th position of $\tau$ is $s_{i+1}$, while the $i$-th position of $\tau'$ is $s_i$. Since $s_i<s_{i+1}$ and $i$ is odd, we have $\tau<_{\mathrm{alt}}\tau'$.
    If $i>j$, then $j$ is the first place in which $\tau$ and $\tau'$ differ. The $j$-th position of $\tau$ is $s_{j}$, while the $j$-th position of $\tau'$ is $s_{j+1}$. Since $s_j<s_{j+1}$ and $j$ is even, we have $\tau<_{\mathrm{alt}}\tau'$ as well. Therefore, the elementary relation $\tau\lessdot\tau'$ always implies $\tau<_{\mathrm{alt}}\tau'$.

    If there is a sequence $\tau\lessdot \tau_1 \lessdot \cdots \lessdot \tau_l \lessdot \tau^{\prime}$, then we have the corresponding sequence $\tau <_{\mathrm{alt}} \tau_1 <_{\mathrm{alt}} \cdots <_{\mathrm{alt}} \tau_l <_{\mathrm{alt}} \tau^{\prime}$. In particular $\tau^{\prime}\nless\tau$, and hence $<$ is asymmetric.

    The irreflexivity and transitivity of $<$ are immediate from the construction.
\end{proof}                                                                                     

 We also use the following lemma:
\begin{lem}\label{lem:absnegative<abspositive}
    Every absolutely positive $(N-2)$-simplex of $\Delta^N$ is greater than every absolutely negative $(N-2)$-simplex of $\Delta^N$ in the order $<$.
\end{lem}

\begin{proof}
    Let $\tau$ be an absolutely negative $(N-2)$-simplex and $\tau'$ be an absolutely positive $(N-2)$-simplex.
    By Proposition~\ref{prop:absolutely_positeve/negative}, $\tau$ is of the form $\tau=[0,\ldots,N]\backslash\{i<j\}$ with odd $i$ and even $j$, and $\tau'$ is of the form $\tau'=[0,\ldots,N]\backslash \{k<l\}$ with even $k$ and odd $l$.

    If $\{i,j\}\cap\{k,l\}\ne\emptyset$ (equivalently $i=l$ or $j=k$), then there is an $(N-1)$-simplex $\sigma$ that contains both $\tau$ and $\tau'$. Explicitly, $\sigma=[0,\ldots,N]\backslash\{i\}$ if $i=l$, and $\sigma=[0,\ldots,N]\backslash\{j\}$ if $j=k$. In particular, $\tau\lessdot\tau'$.

    If $\{i,j\}\cap\{k,l\}=\emptyset$, there are two cases. In the case $k<j$, set $\tau''\coloneqq[0,\ldots,N]\backslash\{k,j\}$. Then we can check $\tau\lessdot\tau''\lessdot\tau'$. In the case $k>j$, set  $\tau''\coloneqq[0,\ldots,N]\backslash\{i,l\}$. Then we can check $\tau\lessdot\tau''\lessdot\tau'$.
\end{proof}

Let $T$ be a set-theoretic solution of the $n$-gon equation and $S$ a set-theoretic solution of the dual $n$-gon equation.
The following proposition shows that $T$- and $S$-colorings are determined by colors of the absolutely negative and positive simplices, respectively.

\begin{prop}\label{prop:coloring_of_absolutely_positive/negative}
    For $N>n-3$, the maps
    \begin{align*}
        \operatorname{Col}_T(\Delta^N) \longrightarrow \operatorname{Map}(\mathcal{N}_{n-3}^N, X) \quad\text{and}\quad \operatorname{Col}_S(\Delta^N) \longrightarrow \operatorname{Map}(\mathcal{P}_{n-3}^N, X)
    \end{align*}
    induced by $\mathcal{N}_{n-3}^N\subset \Delta^N_{n-3}$ and $\mathcal{P}_{n-3}^N\subset \Delta^N_{n-3}$, respectively, are bijective. 
    
    In other words, a $T$-coloring $c \colon \Delta^N_{n-3} \to X$ (resp. $S$-coloring $c \colon \Delta^N_{n-3} \to X$) is determined by colors of the absolutely negative (resp. absolutely positive) $(n-3)$-simplices.
\end{prop}

\begin{proof}

The proof essentially follows \cite{korepanov2016cohomologies}. 
Let $T\colon X^{\times\lfloor\frac{n-1}{2}\rfloor}\to X^{\times\lfloor\frac{n}{2}\rfloor}$ be a solution of the $n$-gon equation, and write it as
\begin{equation*}
    T(x_1, \cdots x_{\lfloor\frac{n-1}{2}\rfloor}) = (T_1(x_1, \cdots x_{\lfloor\frac{n-1}{2}\rfloor}), \cdots, T_{\lfloor\frac{n}{2}\rfloor}(x_1, \cdots x_{\lfloor\frac{n-1}{2}\rfloor})),
\end{equation*}
where $T_i\colon X^{\times\lfloor\frac{n-1}{2}\rfloor}\to X$.

    Given a lower subset\footnote{A subset $A$ such that whenever $\tau<\tau^{\prime}$ and $\tau^{\prime}\in A$, one has $\tau\in A$.} $A\subset \Delta^N_{n-3}$, let $\operatorname{Col}_T(A,X)\subset \operatorname{Map}(A,X)$ denote the set of partial $T$-colorings, i.e., maps $g\colon A \to X$ such that 
    \[
    g(\partial_{2i}\sigma)=T_{i+1}\bigl(g(\partial_1\sigma),g(\partial_3\sigma),\ldots,g(\partial_{2\lfloor\frac{n-1}{2}\rfloor-1}\sigma)\bigr)
    \]
    whenever $\sigma\in\Delta^N_{n-2}$ and $\partial_{2i}\sigma\in A$.
    The expression on the right-hand side is always defined because every negative face of
    $\sigma$ is smaller than
    $\partial_{2i}\sigma$, and $A$ is lower.
    In particular,
    \[
    \operatorname{Col}_T(\Delta^N_{n-3},X)=\operatorname{Col}_T(\Delta^N).
    \]

    Given $f\in \operatorname{Map}(\mathcal{N}_{n-3}^N, X)$, let us construct a $T$-coloring extending $f$.
    For $\tau \in \Delta^N_{n-3}$, we set $(\Delta^N_{n-3})_{<\tau} \coloneq \{\,\tau^{\prime}\in\Delta^N_{n-3}\,|\,\tau^{\prime} < \tau\,\}$, and then define 
    \[
    \Phi(\tau)\colon \operatorname{Col}((\Delta^N_{n-3})_{<\tau},X) \to X
    \] as follows.
    First, note that $(\Delta^N_{n-3})_{<\tau} = \emptyset$ if and only if $\tau \in \mathcal{N}^N_{n-3}$.
    Thus, for $\tau \in \mathcal{N}^N_{n-3}$, $\operatorname{Col}_T((\Delta^N_{n-3})_{<\tau}, X)=\{*\}$, and so we define 
    \[
    \Phi(\tau)(*)\coloneqq f(\tau).
    \]
    For $\tau \in \Delta^N_{n-3}\backslash \mathcal{N}^N_{n-3}$ and $g \in \operatorname{Col}_T((\Delta^N_{n-3})_{<\tau}, X)$, we take $\sigma \in \Delta^{N}_{n-2}$ such that $\tau$ is positive with respect to $\sigma$.
    Since all the $(n-3)$-simplices negative with respect to this $\sigma$ belong to $(\Delta^N_{n-3})_{<\tau}$, we can apply the solution $T$ to get a color of $\tau$, and then define $\Phi(\tau)(g)$ as this color. Formally, if $\tau = \partial_{2i}\sigma$ for some $i\geq 0$, then 
    \[
    \Phi(\tau)(g)\coloneq T_{i+1}(g(\partial_1\sigma), g(\partial_3\sigma), \ldots,g(\partial_{2\lfloor\frac{n-1}{2}\rfloor-1}\sigma)).
    \]

    In order for $\Phi(\tau)(g)$ to be well-defined, we need to show that this does not depend on the chosen $\sigma$.
    Let $\sigma^{\prime}$ be another $(n-2)$-simplex with respect to which $\tau$ is positive. There is a unique $(n-1)$-simplex $\rho$ that contains both $\sigma$ and $\sigma'$. Then $\tau$ is an absolutely positive $(n-3)$-simplex of $\rho\cong\Delta^{n-1}$. So, by Lemma~\ref{lem:absnegative<abspositive}, all absolutely negative $(n-3)$-simplices belong to $(\Delta^N_{n-3})_{<\tau}$. Since $T$ is a solution of the $n$-gon equation, $\sigma$ and $\sigma'$ occur on opposite sides of the $n$-gon equation for $\rho$, and $g$ is a partial $T$-coloring on $(\Delta^N_{n-3})_{<\tau}$, two colors of $\tau$ induced by $\sigma$ and $\sigma^{\prime}$ coincide.

    Now, using recursion and induction, there is a unique $h\in\operatorname{Col}_T(\Delta^N)$ such that for $\tau\in\Delta^N_{n-3}$, $h(\tau)=\Phi(\tau)(h|_{(\Delta^N_{n-3})_{< \tau}})$.
    In particular, $h(\tau)=f(\tau)$ for $\tau\in\mathcal{N}^N_{n-3}$, thus extending $f\in\operatorname{Map}(\mathcal{N}^N_{n-3},X)$.
    
\end{proof}

The set $\Delta^N_{n-3}$ admits a linear order using reverse lexicographical order, i.e.,
for $\tau=[v_0,\ldots,v_{n-3}], \tau^{\prime}=[v_0^{\prime},\ldots,v_{n-3}^{\prime}]\in\Delta^N_{n-3}$, $\tau<_{\text{rex}} \tau^{\prime}$ if and only if there exist $0\leq j\leq n-3$ such that $v_i=v_i^{\prime}$ for $i<j$ and $v_j > v_{j}^{\prime}$.
For the rest of the paper, we fix linear orders on $\mathcal{P}^N_{n-3}$ and $\mathcal{N}^N_{n-3}$ by restricting the reverse lexicographical order of $\Delta^N_{n-3}$.
The set $\{ f\colon \mathcal{N}_{n-3}^N\to X \}$ is then canonically identified with $X^{\times \abs{\mathcal{N}^N_{n-3}}}=\{\,(x_1, x_2, \ldots, x_{\abs{\mathcal{N}^N_{n-3}}})\mid x_i\in X\,\}$, where $x_i$ is the color of the $i$'th absolutely positive/negative $(n-3)$-simplex of $\Delta^N$.  

Let $T$ be a solution of the $n$-gon equation, and $S$ a solution of the dual $n$-gon equation.
Then, for $N\geq n-3$, the chain groups have the form
\begin{equation*}
    C_N(T) = \mathbb{Z}\langle (x_1, x_2, \ldots, x_{\abs{\mathcal{N}^N_{n-3}}})\mid x_i\in X \rangle, \qquad C_N(S) = \mathbb{Z}\langle (x_1, x_2, \ldots, x_{\abs{\mathcal{P}^N_{n-3}}})\mid x_i\in X \rangle.
\end{equation*}

\section{Examples}\label{sec:examples}

\subsection{$H_{n-3}$, $H^{n-3}$, and the structure group}

Let $T\colon X^{\times\lfloor\frac{n-1}{2}\rfloor}\to X^{\times\lfloor\frac{n}{2}\rfloor}$ be a solution of the $n$-gon equation, and write it as
\begin{equation*}
    T(x_1, \cdots x_{\lfloor\frac{n-1}{2}\rfloor}) = (T_1(x_1, \cdots x_{\lfloor\frac{n-1}{2}\rfloor}), \cdots, T_{\lfloor\frac{n}{2}\rfloor}(x_1, \cdots x_{\lfloor\frac{n-1}{2}\rfloor})),
\end{equation*}
where $T_i\colon X^{\times\lfloor\frac{n-1}{2}\rfloor}\to X$.
Then the $(n-3)$-th and $(n-2)$-th chain groups are 
\begin{equation*}
    C_{n-3}(T)=\mathbb{Z}\langle (x)\mid x\in X\rangle\quad\text{and}\quad 
    C_{n-2}(T)=\mathbb{Z}\langle (x_1,\ldots,x_{{\lfloor\frac{n-1}{2}\rfloor}})\mid x_i\in X\rangle,
\end{equation*}
respectively, and the boundary map $\partial_{n-2} \colon C_{n-2}(T) \to C_{n-3}(T)$  is
\begin{equation*}
 \partial_{n-2}(x_1, x_2, \ldots, x_{{\lfloor\frac{n-1}{2}\rfloor}}) = \sum_{i=1}^{{\lfloor\frac{n}{2}\rfloor}} T_i(x_1, \cdots ,x_{{\lfloor\frac{n-1}{2}\rfloor}}) - \sum_{i=1}^{{\lfloor\frac{n-1}{2}\rfloor}} x_i.
\end{equation*}
Thus, if we define the \textbf{structure group} $G_T$ of $T$ by 
\begin{equation*}
    G_T \coloneqq \langle x\in X \mid x_1\cdots x_{{\lfloor\frac{n-1}{2}\rfloor}} = T_1(x_1,\cdots,x_{{\lfloor\frac{n-1}{2}\rfloor}}) \cdots T_{\lfloor \frac{n}{2} \rfloor}(x_1,\cdots,x_{{\lfloor\frac{n-1}{2}\rfloor}})\rangle,
\end{equation*}
then
\begin{equation*}
    H_{n-3}(T) \cong (G_T)_{\mathrm{ab}}\quad\text{and}\quad H^{n-3}(T) \cong \operatorname{Hom}(G_T,\mathbb{Z}).
\end{equation*}

\subsection{Cohomology of the (dual) $n$-gon equations for small $n$}\label{subsec:cohomologyforsmalln}
\subsection*{Cohomology of the $4$-gon equation}
Recall that a set-theoretic solution of the 4-gon equation is a map 
\begin{equation*}
    \Delta\colon X\longrightarrow X\times X
\end{equation*}
satisfying $(\Delta\times \operatorname{id}_X)\circ \Delta = (\operatorname{id}_X\times \Delta)\circ \Delta$. Let us write $\Delta(x)=(\Delta_1(x),\Delta_2(x))$ where $\Delta_1,\Delta_2\colon X \to X$.

For each $N \geq 2$, there is only one absolutely negative simplex $\mathcal{N}^N_1 = \{[0,N]\}$, 
and thus the chain complex is given by 
\begin{equation*}
    C_N(\Delta) = \mathbb{Z} \langle (x) \mid x\in X \rangle 
\end{equation*}
for $N\geq 1$.
Here, each basis elemet $(x)$ of $C_N(\Delta)$ corresponds to the $\Delta$-coloring of $\Delta^N$ defined by 
\[
c_x\colon \mathcal{N}^N_1\longrightarrow X; \qquad [0,N]\longmapsto x.
\]

Now, let us compute the boundary map $\partial_{N} \colon C_N(\Delta) \to C_{N-1}(\Delta)$ for $N\geq 2$.

The coloring $d_i^N (c_x)$ is determined by a color of the absolutely negative edge $[0,N-1]$ of $\Delta^{N-1}$, and we calculate
\begin{align*}
    d_i^N(c_x)([0,N-1]) = c_x(\delta^i([0,N-1])) =
    \begin{cases}
        c_x([1,N]) & \text{if $i=0$,}\\
        c_x([0,N])=x & \text{if $0<i<N$,}\\
        c_x([0,N-1]) & \text{if $i=N$.}
    \end{cases}
\end{align*}
Since $[1,N]=\partial_0 [0,1,N]$, the $\Delta$-coloring rule gives $c_x([1,N])=\Delta_1(x)$, and thus $d_0^N((x))=(\Delta_1(x))$.
Likewise, the $\Delta$-coloring rule for $[0,N-1,N]$ gives $d_N^N((x))=(\Delta_2(x))$.
The boundary map for $N \geq 2$ is therefore
\begin{equation*}
    \partial_N (x)=
    \begin{cases}
        (\Delta_1(x))-(\Delta_2(x)) & \text{if $N$ is odd,} \\
        (\Delta_1(x))-(x)+(\Delta_2(x)) & \text{if $N$ is even.} 
    \end{cases}
\end{equation*}
Now, if we take the dual and compute the cohomology, this results in
\begin{align*}
    H^{1}(\Delta)=\{\,\phi\colon X\to \mathbb{Z}\mid \phi(x)=\phi(\Delta_1(x))+\phi(\Delta_2(x))\,\}\quad\text{and}\quad H^{N}(\Delta)=0 \quad \text{for $N \ne 1$}.
\end{align*}

\subsection*{Cohomology of the dual $4$-gon equation}
A set-theoretic solution of the dual 4-gon equation is a map 
\begin{equation*}
    M\colon X\times X\to X
\end{equation*}
satisfying $M\circ (M\times \operatorname{id}_X) = M\circ(\operatorname{id}_X\times M)$, that is, $(X,M)$ is a semigroup.
For simplicity, we write $x y\coloneqq M(x,y)$.

The absolutely positive $1$-simplices are given by 
\begin{equation*}
    \mathcal{P}^N_1=\{[N-1,N]<[N-2,N-1]<\cdots<[1,2]<[0,1]\},
\end{equation*}
where the linear order was given by reverse lexicographical order.
The chain complex is thus
\begin{equation*}
    C_N(M)=\mathbb{Z} \langle (x_1,x_2,\ldots,x_N)\mid x_i\in X\rangle,
\end{equation*}
where each generator $x=(x_1,\ldots,x_N)$ corresponds to the $M$-coloring $c_x$ of $\Delta^N$ defined by $\mathcal{P}^N_1\ni[j,j+1]\mapsto x_{N-j}\in X$. The coloring $d_i^N(c_x)$ colors the absolutely positive edge $[j,j+1]$ of $\Delta^{N-1}$ ($0\leq j \leq N-2$) with
\begin{equation*}
    d_i^N(c_x)([j,j+1]) = c_x(\delta^i[j,j+1]) =
    \begin{cases}
        c_x([j+1,j+2]) = x_{N-j-1} & \text{if $ i < j+1$,} \\
        c_x([j,j+2]) & \text{if $i=j+1$,} \\
        c_x([j,j+1]) = x_{N-j} & \text{if $i>j+1$.} \\
    \end{cases}
\end{equation*}
Since $[j,j+2] = \partial_1 [j, j+1, j+2]$, the $M$-coloring rule gives $c_x([j,j+2])=x_{N-j-1}x_{N-j}$.
Thus we have
\begin{equation*}
    d_i^N(x_1,x_2,\ldots, x_N) = 
    \begin{cases}
    (x_1, \ldots , x_{N-1}) & \text{if $i=0$},\\
    (x_1,\ldots, x_{N-i-1},x_{N-i}x_{N-i+1},x_{N-i+2},\ldots ,x_{N}) & \text{if $0 < i < N$},\\
    (x_2, \ldots , x_N)& \text{if $i=N$},
    \end{cases}
\end{equation*}
and the boundary map is
\begin{equation*}
    \partial_N(x_1,\ldots,x_N) = (-1)^N\left((x_2,\ldots,x_N) + \sum_{i=1}^{N-1}(-1)^i(x_1,\ldots,x_ix_{i+1},\ldots, x_N) +(-1)^N(x_1,\ldots,x_{N-1}) \right)
\end{equation*}
if $N\geq2$, and $0$ otherwise. 

Therefore, the chain complex coincides, up to sign, with the usual bar chain complex of the semigroup $(X,M)$ in degrees $N\geq2$.
Also, the semi-simplicial set $\operatorname{Col}_M$ is isomorphic to $BX=B(*,X,*)$, the bar complex, as a semi-simplicial set in degrees $N\geq 2$.

\subsection*{Cohomology of the $5$-gon equation}
A set-theoretic solution of the 5-gon equation is a map 
\begin{equation*}
    T\colon X\times X\to X\times X
\end{equation*}
satisfying $T_{12}T_{13}T_{23}=T_{23}T_{12}$.
For simplicity, we write $T(x,y)=(x\circ y,x\cdot y)$, where $ \circ, \cdot \colon X\times X\to X$.
For $N\geq 2$, the chain complex is given by 
\begin{equation*}
    C_{N}(T)=\mathbb{Z}\langle (x_1,\ldots ,x_{N-1})\mid x_i\in X \rangle,
\end{equation*}
where each generator $x=(x_1,\ldots,x_{N-1})$ corresponds to the $T$-coloring $c_x$ of $\Delta^N$ defined by $\mathcal{N}^N_2\ni[0,j,j+1]\mapsto x_{N-j
}\in X$. For $1\leq j \leq N-2$, the coloring of the absolutely negative 2-simplex $[0,j,j+1]$ of $\Delta^{N-1}$ is
\begin{align*}
d_i^N(c_x)([0,j,j+1])=
    \begin{cases}
        c_x([1,j+1,j+2]) & \text{if $i=0$},\\
        c_x([0,j+1,j+2]) = x_{N-j-1} & \text{if $0<i<j+1$},\\
        c_x([0,j,j+2]) = x_{N-j-1} \cdot x_{N-j} & \text{if $i=j+1$},\\
        c_x([0,j,j+1]) = x_{N-j} & \text{if $j+1<i$}.\\
    \end{cases}
\end{align*}
In order to determine the coloring $c_x([1,j+1,j+2])$, consider the 3-simplices $\{\,[0,1,m,m+1]\,\}_{2 \leq m \leq j+1}$.
Notice that $[1,j+1,j+2] = \partial_0 [0,1,j+1,j+2]$ and $\partial_3 [0,1,m,m+1]=\partial_2[0,1,m-1,m]$.
Thus, starting from the 3-simplex $[0,1,2,3]$, we can apply the map $T$ consecutively to get the coloring of $[1,j+1,j+2]$, which results in $c_x([1,j+1,j+2])=A_{N-j-1}$, where $A_k=x_k \circ (x_{k+1}\cdot x_{k+2}\cdot\ldots\cdot \cdot x_{N-1})$ and $A_{N-2}=x_{N-2}\circ x_{N-1}$.
Note that the product $x\cdot y$ is always associative.

Thus, the face map is given by
\begin{align*}
    d_i^N(x_1,\ldots, x_{N-1})=
    \begin{cases}
        (A_1, \ldots, A_{N-2}) & \text{if $i=0$},\\
        (x_1, \ldots ,x_{N-2}) & \text{if $i=1$},\\
        (x_1,\ldots, x_{N-i-1},\,x_{N-i}\cdot x_{N-i+1}\,,x_{N-i+2},\ldots ,x_{N-1}) & \text{if $1 < i < N$},\\
        (x_2, \dots ,x_{N-1}) & \text{if $i = N$},\\
    \end{cases}
\end{align*}
and the boundary map $\partial_{N}$ is given by
\begin{equation*}
  \begin{split}
    \partial_{N}(x_1,x_2,\ldots x_{N-1})
      = {} & (-1)^{N}(x_2,\ldots,x_{N-1})+\sum_{k=1}^{N-2}(-1)^{N+k}(x_1,\ldots,x_{k-1},x_k\cdot x_{k+1},x_{k+2},\ldots, x_{N-1}) \\
           & - (x_1,\ldots,x_{N-2}) + (A_1,\ldots,A_{N-2}).
  \end{split}
\end{equation*}

Note that this semi-simplicial set $\operatorname{Col}_T$ appears as a special case of 2-Segal semi-simple sets in \cite[Chapter 3.7]{dyckerhoff2019higher}, where $\operatorname{Col}_T$ is called the nerve of $(X, T)$.

\begin{ex}
    Let $X$ be a group. Consider the solution $T:X\times X\to X\times X$ of the pentagon equation given by Eq.~\eqref{eq:Tcommutativemonoid}: $T(x,y) \coloneqq (x,xy)$.
    In this case, the face map is
    \begin{align*}
        d_i^N(x_1,\ldots, x_{N-1})=
        \begin{cases}
            (x_1, \ldots ,x_{N-2}) & \text{if $i=0,1$},\\
            (x_1,\ldots, x_{N-i-1},\,x_{N-i}x_{N-i+1},\,x_{N-i+2},\ldots ,x_{N-1}) & \text{if $1 < i < N$},\\
            (x_2, \dots ,x_{N-1}) & \text{if $i = N$},\\
        \end{cases}
    \end{align*}
    and hence, the boundary map is
    \begin{equation*}
        \partial_{N}(x_1,x_2,\ldots x_{N-1})
        = (-1)^{N}\left((x_2,\ldots,x_{N-1})+\sum_{k=1}^{N-2}(-1)^k(x_1,\ldots,x_{k-1},x_k x_{k+1},x_{k+2},\ldots, x_{N-1})\right).
\end{equation*}
So, if we define a new semi-simplicial set $(\tilde{C}_N,\tilde{d}_N)$ by
\begin{equation*}
    \tilde{C}_N \coloneqq (\operatorname{Col}_T)_{N+2},\quad \tilde{d}^N_i \coloneqq d_{i+2}^{N+2}
\end{equation*}
i.e., the two-fold d{\'e}calage of $\operatorname{Col}_T$ (see for example \cite{danny2012decalage}), then it is again isomorphic to $EX=B(*,X,X)$, the two-sided bar complex, as a semi-simplicial set.
\end{ex}

\begin{ex}
    Let $\mathbb{Z}_{m}$ be the cyclic group of order $m$ and consider the solution $T\colon \mathbb{Z}_{m} \times \mathbb{Z}_{m} \to \mathbb{Z}_{m} \times \mathbb{Z}_{m}$ of the pentagon equation given by $T(x,y)\coloneq(0,x+y)$. In this case, the face map is
    \begin{align*}
        d_i^N(x_1,\ldots, x_{N-1})=
        \begin{cases}
            (0,\ldots,0) & \text{if $i=0$},\\
            (x_1, \ldots ,x_{N-2}) & \text{if $i=1$},\\
            (x_1,\ldots, x_{N-i-1},\,x_{N-i}+x_{N-i+1},\,x_{N-i+2},\ldots ,x_{N-1}) & \text{if $1 < i < N$},\\
            (x_2, \dots ,x_{N-1}) & \text{if $i = N$}.\\
        \end{cases}
    \end{align*}
    Then the homology and cohomology of the semi-simplicial set are
    \begin{equation*}
    H_{N}(T;\mathbb{Z}) \cong 
        \begin{cases}
            0 & \text{if $N=0,1$},\\
            \mathbb{Z}\oplus\mathbb{Z}_m & \text{if $N=2$},\\
            0 & \text{if $N\ge 3$ is odd},\\
            \mathbb{Z}_m & \text{if $N\ge 4$ is even,}
        \end{cases}\qquad
    H^{N}(T;\mathbb{Z}) \cong 
        \begin{cases}
            0 & \text{if $N=0,1$},\\
            \mathbb{Z} & \text{if $N=2$},\\
            \mathbb{Z}_m & \text{if $N\ge3$ is odd},\\
            0 & \text{if $N=1$ or $N\ge4$ is even}.
        \end{cases}
    \end{equation*}
    Indeed,
    \begin{equation*}
        \partial_3(a,b) = (0)-(a)+(a+b)-(b),
    \end{equation*}
    so $H_2(T)=\mathbb{Z}[\mathbb{Z}_m]/\operatorname{Im}(\partial_3)\cong\mathbb{Z}\oplus\mathbb{Z}_m$. For $N\ge3$, consider the subcomplex $E_\bullet$ of $C_\bullet(T)$ generated by the zero tuples. The quotient $C_\bullet(T)/E_\bullet$ identifies, up to a degree shift and signs, with the bar complex modulo the constant zero tuples. Hence, $H_N(T)\cong H_{N-1}(\mathbb{Z}_m;\mathbb{Z})$ $(N\ge 3)$. The cohomology then follows from the universal coefficient theorem.
\end{ex}

\section{Some properties of polygon cohomology}\label{sec:properties}

\subsection{Cohomology for a bijective solution}
Let $T\colon X^{\times k}\to X^{\times k}$ be a bijective map. Then $T$ is a solution of the $(2k+1)$-gon equation if and only if the inverse map $\overline{T}$ is a solution of the dual $(2k+1)$-gon equation.

\begin{prop}\label{prop:cohomologybijective}
    For a bijective solution $T$ of the $(2k+1)$-gon equation, 
    the semi-simplicial sets $\operatorname{Col}_T$ and $\operatorname{Col}_{\overline{T}}$ coincide.
\end{prop}
\begin{proof}
    It is obvious that a $T$-coloring of $\Delta^N$ is also a $\overline{T}$-coloring, and vice versa. Thus, $\operatorname{Col}_{T}(\Delta^N)=\operatorname{Col}_{\overline{T}}(\Delta^N)$.
    Since the definitions of the face maps are independent of $T$ and $\overline{T}$, the result follows.
\end{proof}

\subsection{Vanishing of (co)homology}\label{subsec:vanishingcohomology}

Let $X$ be a set, $T^{(m)}\colon X^{\times\lfloor\frac{m-1}{2}\rfloor}\to X^{\times\lfloor\frac{m}{2}\rfloor}$ be a solution of the $m$-gon equation, and $S^{(m)}\colon X^{\times\lfloor\frac{m}{2}\rfloor}\to X^{\times\lfloor\frac{m-1}{2}\rfloor}$ be a solution of the dual $m$-gon equation. In \cite[Corollary~4.2]{mihalache2025constructing}, we showed that the maps $\Delta \circ_{r} S^{(m)}$, $T^{(2k+1)}\circ_{l}\Delta$, and $S^{(2k)}\circ_{l}\Delta$ defined by
\begin{align*}
    \Delta \circ_{r} S^{(m)} (x_1, x_2, \ldots, x_{\lfloor\frac{m}{2}\rfloor}) &\coloneqq (x_1, S^{(m)}(x_1, x_2, \ldots x_{\lfloor{\frac{m}{2}\rfloor}})),\\
    T^{(2k+1)} \circ_{l} \Delta (x_1, x_2, \ldots, x_k) &\coloneqq (T^{(2k+1)}(x_1, x_2, \ldots ,x_k), x_k),\\
    S^{(2k)} \circ_{l} \Delta(x_1,x_2,\ldots,x_k) &\coloneqq (S^{(2k)}(x_1, x_2, \ldots ,x_k), x_k)
\end{align*}
are solutions of the $(m+1)$-gon, $(2k+2)$-gon, and dual $(2k+1)$-gon equations, respectively.

\begin{prop}\label{prop:cohomologystacking}
The following hold:
\begin{enumerate}[label=(\arabic*)]
    \item If there exists an element $e \in X$ such that $S^{(m)}(x, e, \ldots , e)=(x, e, \ldots, e)$ for all $x\in X$, then the $N$-th polygon (co)homology of $\Delta \circ_{r} S^{(m)}$ vanishes for all $N \geq m-1$.\label{prop:cohomologystacking1}
    \item If there exists an element $e \in X$ such that $T^{(2k+1)}(e,\ldots,e,x)=(e,\ldots,e,x)$ for all $x\in X$, then the $N$-th polygon (co)homology of $T^{(2k+1)}\circ_{l}\Delta$ vanishes for all $N \geq 2k$.\label{prop:cohomologystacking2}
    \item If there exists an element $e\in X$ such that $S^{(2k)}(e,\ldots,e,x)=(e,\ldots,e,x)$ for all $x\in X$, then the $N$-th polygon (co)homology of $S^{(2k)} \circ_{l} \Delta$ vanishes for all $N\geq 2k-1$.\label{prop:cohomologystacking3}
\end{enumerate}
\end{prop}

\begin{proof}
    We prove the proposition by directly constructing a contracting homotopy. 

    For \ref{prop:cohomologystacking1}, set, for simplicity, $T\coloneqq\Delta \circ_{r} S^{(m)}$ and $n\coloneqq m+1$. For each $N\geq n-3$, we define a contracting homotopy $h_N\colon C_N(T)\to C_{N+1}(T)$ as follows: For any map $c\colon\Delta_{n-3}^N\to X$ and an $(n-3)$-simplex $[v_0,v_1,\ldots,v_{n-3}]\subset \Delta^{N+1}$, set
    \begin{equation*}
        (h_N(c))([v_0,v_1,\ldots,v_{n-3}]) \coloneqq \left\{\begin{array}{ll}
            c([v_0-1,v_1-1,\ldots,v_{n-3}-1]) & \text{if $v_0\ne 0$,} \\[2pt]
            c([0,v_1-1, v_2-1, \ldots,v_{n-3}-1]) & \text{if $v_0=0, v_1\ne1$,} \\[2pt]
            c([0,1,v_2-1,\ldots,v_{n-3}-1]) & \text{if $v_0=0, v_1=1, v_2\ne2$,} \\[2pt]
            \multicolumn{1}{c}{e} & \text{if $v_0=0, v_1=1, v_2=2$.}
        \end{array}
        \right.
    \end{equation*}
    We then have the following two lemmas:
    \begin{lem}
        If $c\in \operatorname{Col}_{T}(\Delta^N)$, then $h_N(c)\in \operatorname{Col}_{T}(\Delta^{N+1})$. Hence, we obtain a well-defined map $h_N\colon C_N(T)\to C_{N+1}(T)$.
    \end{lem}

    \begin{proof}
        Let $\sigma=[v_0,\ldots,v_{n-2}]\subset\Delta^{N+1}$ be an $(n-2)$-simplex. We look at the colors of the $(n-3)$-faces $\partial_i\sigma$ by $h_N(c)$ in the following four cases:
        \begin{center}
            \begin{enumerate*}[label=(\roman*), itemjoin={\quad}]
                \item $v_0\ne0$, \label{case:i}
                \item $v_0=0, v_1\ne 1$, \label{case:ii}
                \item $v_0=0, v_1=1, v_2\ne 2$, \label{case:iii}
                \item $v_0=0, v_1=1, v_2=2$.\label{case:iv}
            \end{enumerate*}
        \end{center}
        
        In case \ref{case:i}, set the $(n-2)$-simplex $\sigma^{\prime}\coloneqq[v_0-1,v_1-1,\ldots,v_{n-2}-1]\subset\Delta^N$. By the definition of $h_N$, $(h_N(c))(\partial_i\sigma)=c(\partial_i\sigma^{\prime})$ for all $0\leq i\leq n-2$. Therefore,
        \begin{align*}
            \MoveEqLeft T((h_N(c))(\partial_1\sigma),(h_N(c))(\partial_3\sigma),\ldots,(h_N(c))(\partial_{2\lfloor\frac{n-1}{2}\rfloor-1}\sigma))\\
            &= T(c(\partial_1\sigma^{\prime}),c(\partial_3\sigma^{\prime}),\ldots,c(\partial_{2\lfloor\frac{n-1}{2}\rfloor-1}\sigma^{\prime}))\\
            &= (c(\partial_0\sigma^{\prime}),c(\partial_2\sigma^{\prime}),\ldots,c(\partial_{2\lfloor\frac{n}{2}\rfloor-2}\sigma^{\prime}))\\
            &= ((h_N(c))(\partial_0\sigma),(h_N(c))(\partial_2\sigma),\ldots,(h_N(c))(\partial_{2\lfloor\frac{n}{2}\rfloor-2}\sigma)).
        \end{align*}

        In case \ref{case:ii}, set the $(n-2)$-simplex $\sigma^{\prime}\coloneqq[0,v_1-1,\ldots,v_{n-2}-1]\subset\Delta^N$. Then the rest is the same as case \ref{case:i}.

        In case \ref{case:iii}, set the $(n-2)$-simplex $\sigma^{\prime}\coloneqq[0,1,v_2-1,\ldots,v_{n-2}-1]\subset\Delta^N$. Then the rest is the same as case \ref{case:i}.

        In case \ref{case:iv}, by the definition of $h_N$, the colors are
        \begin{align*}
            (h_N(c))(\partial_0\sigma) &= (h_N(c))(\partial_1\sigma) = (h_N(c))(\partial_2\sigma) = c([0,1,v_3-1,\ldots,v_{n-2}-1]) \eqqcolon c_0,\\
            (h_N(c))(\partial_i\sigma) &= e\quad\text{for $3\leq i\leq n-2$.}
        \end{align*}
        We then have
        \begin{align*}
            \MoveEqLeft T((h_N(c))(\partial_1\sigma),(h_N(c))(\partial_3\sigma),\ldots,(h_N(c))(\partial_{2\lfloor\frac{n-1}{2}\rfloor-1}\sigma))\\
            &= ((h_N(c))(\partial_1\sigma),S^{(m)}((h_N(c))(\partial_1\sigma),(h_N(c))(\partial_3\sigma),\ldots,(h_N(c))(\partial_{2\lfloor\frac{n-1}{2}\rfloor-1}\sigma)))\\
            &= (c_0,S^{(m)}(c_0,e,\ldots,e))\\
            &= (c_0,c_0,e,\ldots,e)\\
            &= ((h_N(c))(\partial_0\sigma),(h_N(c))(\partial_2\sigma),(h_N(c))(\partial_4\sigma),\ldots,(h_N(c))(\partial_{2\lfloor\frac{n}{2}\rfloor-2}\sigma)).
        \end{align*}
        Here, the third equality follows from the assumption $S^{(m)}(x,e,\ldots,e)=(x,e,\ldots,e)$ for all $x\in X$.

        In conclusion, we have $h_N(c)\in\operatorname{Col}_T(\Delta^{N+1})$.
    \end{proof}
        
    \begin{lem}\label{lem:contractinghomotopy}
    Let $N>n-3$.
        \begin{enumerate}[label=(\arabic*)]
            \item For any $1\leq i\leq N+1$, $d_i^{N+1}\circ h_N = h_{N-1}\circ d_{i-1}^N$ as maps $C_N(T)\to C_N(T)$.\label{lem:contractinghomotopy1}
            \item $d_0^{N+1}\circ h_N=\mathrm{id}_{C_N(T)}$.\label{lem:contractinghomotopy2}
        \end{enumerate}
    \end{lem}

    \begin{proof}
        Take a coloring $c\in \operatorname{Col}_{T}(\Delta^N)$ and an $(n-3)$-simplex $[v_0,\ldots,v_{n-3}]\subset\Delta^N$. 
        
        For \ref{lem:contractinghomotopy1}, let $k$ be such that $v_{k-1}<i\leq v_k$. We can then easily calculate that
        \begin{align*}
            \MoveEqLeft  (d_i^{N+1}(h_N(c)))([v_0,\ldots,v_{n-3}]) = (h_{N-1}(d_{i-1}^N(c)))([v_0,\ldots,v_{n-3}])\\[1pt]
            ={}& \left\{\begin{array}{ll}
            c([v_0-1,\ldots,v_{k-1}-1,v_k,\ldots,v_{n-3}]) & \text{if $v_0\ne 0$,} \\[2pt]
            c([1,v_1,\ldots, v_{n-3}]) & \text{if $v_0=0, v_1\ne1$, $i=1$,} \\[2pt]
            c([0,v_1-1,\ldots, v_{k-1}-1, v_k, \ldots,v_{n-3}]) & \text{if $v_0=0, v_1\ne1$, $i\ne1$,} \\[2pt]
            c([1,2,v_2,\ldots,v_{n-3}]) & \text{if $v_0=0, v_1=1, v_2\ne2$, $i=1$,} \\[2pt]
            c([0,2,v_2,\ldots,v_{n-3}]) & \text{if $v_0=0, v_1=1, v_2\ne2$, $i=2$,} \\[2pt]
            c([0,1,v_2-1,\ldots,v_{k-1}-1,v_k\ldots,v_{n-3}]) & \text{if $v_0=0, v_1=1, v_2\ne2$, $i\ne 1,2$,} \\[2pt]
            \multicolumn{1}{c}{e} & \text{if $v_0=0, v_1=1, v_2=2$.}
            \end{array}\right.
        \end{align*}

        For \ref{lem:contractinghomotopy2}, we have
        \begin{equation*}
            (d_0^{N+1}(h_N(c)))([v_0,\ldots,v_{n-3}]) = (h_N(c))([v_0+1,\ldots,v_{n-3}+1]) = c([v_0,\ldots,v_{n-3}]).
        \end{equation*}
\end{proof} 
    By Lemma~\ref{lem:contractinghomotopy}, we have
    \begin{align*}
        \MoveEqLeft \partial_{N+1}h_N + h_{N-1}\partial_N\\
        ={}& \sum_{i=0}^{N+1}(-1)^id_i^{N+1}\circ h_N + \sum_{i=0}^N(-1)^ih_{N-1}\circ d_i^N\\
        ={}& d_0^{N+1}\circ h_N=\mathrm{id}_{C_N(T)}
    \end{align*}
    for all $N\geq n-2$, and hence $H_N(T) = 0$ for $N\geq n-2$.

    For \ref{prop:cohomologystacking2} and \ref{prop:cohomologystacking3}, set $(T\coloneqq T^{(2k+1)}\circ_{l}\Delta, n\coloneqq 2k+2)$ or $(S\coloneqq S^{(2k)}\circ_l\Delta, n\coloneqq 2k+1)$. In both cases, a contracting homotopy $h_N\colon C_N(T)\to C_{N+1}(T)$ (or $h_N\colon C_N(S)\to C_{N+1}(S)$) ($N\geq n-3$) is given by
     \begin{equation*}
        (h_N(c))([v_0,\ldots,v_{n-3}]) \coloneqq \left\{\begin{array}{ll}
            c([v_0,\ldots,v_{n-3}]) & \text{if $v_{n-3}\ne N+1$,} \\[2pt]
            c([v_0,\ldots,v_{n-4}, N]) & \text{if $v_{n-3}=N+1, v_{n-4}\ne N$,} \\[2pt]
            c([v_0,\ldots,v_{n-5},N-1,N]) & \text{if $v_{n-3}=N+1, v_{n-4}= N, v_{n-5}\ne N-1$} \\[2pt]
            \multicolumn{1}{c}{e} & \text{if $v_{n-3}=N+1$, $v_{n-4}= N$, $v_{n-5}=N-1$}
        \end{array}
        \right.
    \end{equation*}
    for $c\colon\Delta_{n-3}^N\to X$ and $[v_0,v_1,\ldots,v_{n-3}]\subset \Delta^{N+1}$. Then the rest is similar.
\end{proof}

As an example, we apply Proposition~\ref{prop:cohomologystacking} to the solutions $T^{(n)}$ and $S^{(n)}$ of (dual) polygon equations coming from a commutative monoid in Example~\ref{ex:sol_commmonoid}.

\begin{cor}\label{cor:homologymonoid}
Let $X$ be a commutative monoid, and $T^{(n)}$ and $S^{(n)}$ be the solutions of the $n$-gon and dual $n$-gon equations given by Eq.~\eqref{eq:Tcommutativemonoid} and  \eqref{eq:Scommutativemonoid}, respectively. Then the homology groups are
\begin{equation*}
    H_N(T^{(2k+1)}) = H_N(S^{(2k+1)})=\left\{
    \begin{array}{ll}
        \mathbb{Z} & \text{if $N=2k-2$,}\\[2pt]
        0 & \text{otherwise,}
    \end{array}
        \right.\qquad
    H_N(T^{(2k+2)}) = 0\quad\text{for all $N$.}\footnotemark
\end{equation*}
\footnotetext{The homology groups of dual even-gon equations are, in general, nontrivial.}
\end{cor}

\begin{proof}
    First note that $T^{(2k+1)}=\Delta \circ_{r} S^{(2k)}$, $T^{(2k+2)}=T^{(2k+1)}\circ_l\Delta$, and $S^{(2k+1)}=S^{(2k)}\circ_{l}\Delta$, and also $S^{(2k)}(x,1,\ldots,1)=(x,1,\ldots,1)$, $T^{(2k+1)}(1,\ldots,1,x)=(1,\ldots,1,x)$, and $S^{(2k)}(1,\ldots,1,x)=(1,\ldots,1,x)$ for all $x\in X$, where $1$ is the identity element of $X$. So, by Proposition~\ref{prop:cohomologystacking}, all homology groups are zero except $H_{2k-2}(T^{(2k+1)})$, $H_{2k-2}(S^{(2k+1)})$, and $H_{2k-1}(T^{(2k+2)})$.

    For $H_{2k-2}(T^{(2k+1)})$ (and similarly for $H_{2k-2}(S^{(2k+1)})$), the boundary map $\partial_{2k-1}\colon C_{2k-1}(T^{(2k+1)})\to C_{2k-2}(T^{(2k+1)})$ in terms of the canonical basis is given by
    \begin{align*}
        \partial_{2k-1}(x_1,\ldots, x_k) &= x_1-x_1+x_1x_2-x_2+x_2x_3-x_3+\cdots-x_{k-1}+x_{k-1}x_k-x_k\\
        &= \sum_{i=1}^{k-1}(x_i-1)x_{i+1}.
    \end{align*}
    Thus, the image is contained in the augmentation ideal $I$. Conversely, for any $x,y\in X$,
    \begin{equation*}
        (x-1)y = \partial_{2k-1}(1\ldots, 1,x,y),
    \end{equation*}
    so that $I\subset\operatorname{Im}(\partial_{2k-1})$. Hence,
    \begin{equation*}
        H_{2k-2}(T^{(2k+1)}) = \mathbb{Z}\langle X\rangle/\operatorname{Im}(\partial_{2k-1}) = \mathbb{Z}\langle X\rangle/I \cong \mathbb{Z}.
    \end{equation*}
    
    For $H_{2k-1}(T^{(2k+2)})$, the boundary map $\partial_{2k}\colon C_{2k}(T^{(2k+2)})\to C_{2k-1}(T^{(2k+2)})$ in terms of the canonical basis is given by
    \begin{align*}
        \partial_{2k}(x_1,\ldots, x_k) &= x_1-x_1+x_1x_2-x_2+x_2x_3-x_3+\cdots-x_{k-1}+x_{k-1}x_k-x_k+x_k\\
        &= \sum_{i=1}^{k-1}x_ix_{i+1} - \sum_{i=2}^{k-1}x_i.
    \end{align*}
    In particular, if one takes $x_2=\cdots=x_k=1$ (the identity of $X$), then
    \begin{equation*}
        \partial_{2k}(x_1,1,\ldots, 1) = x_1,
    \end{equation*}
    so that $\partial_{2k}$ is surjective. Hence, $H_{2k-1}(T^{(2k+2)})=0$.
\end{proof}

\subsection{(Co)homology of stacking solutions}

Let $X$ be a set, and $S^{(n)}\colon X^{\times\lfloor\frac{n}{2}\rfloor}\to X^{\times\lfloor\frac{n-1}{2}\rfloor}$ be a solution of the dual $n$-gon equation. As in Section~\ref{subsec:vanishingcohomology}, the map $T^{(n+1)}\coloneqq \Delta\circ_{r}S^{(n)} \colon X^{\times\lfloor\frac{n}{2}\rfloor}\to X^{\times\lfloor\frac{n+1}{2}\rfloor}$ given by
\begin{equation*}
    T^{(n+1)}(x_1,\ldots,x_{\lfloor\frac{n}{2}\rfloor}) \coloneqq (x_1,S^{(n)}(x_1,\ldots, x_{\lfloor\frac{n}{2}\rfloor}))
\end{equation*}
is a solution of the $(n+1)$-gon equation. In this case, one can define maps $\{\pi_{N+2}\colon \operatorname{Col}_{T^{(n+1)}}(\Delta^{N+2})\to \operatorname{Col}_{S^{(n)}}(\Delta^N)\}$ whose linear extension form a chain map of degree $-2$, as follows: For any map $c\colon\Delta_{n-2}^{N+2}\to X$ and an $(n-3)$-simplex $[v_0,v_1,\ldots,v_{n-3}]\subset \Delta^N$, set
\begin{equation*}
    (\pi_{N+2}c)([v_0,v_1,\ldots,v_{n-3}]) \coloneqq c([0,v_0+2, v_1+2,\ldots,v_{n-3}+2]).
\end{equation*}
We then have
\begin{prop}\label{prop:chainmap}
    \begin{enumerate}[label=(\arabic*)]
        \item If $c\in \operatorname{Col}_{T^{(n+1)}}(\Delta^{N+2})$, then $\pi_{N+2}c\in \operatorname{Col}_{S^{(n)}}(\Delta^N)$.\label{prop:chainmap1}
        \item $\{\pi_{N+2}\}$ is a chain map of degree $-2$, i.e., $\partial_N^{S^{(n)}}\pi_{N+2}=\pi_{N+1}\partial_{N+2}^{T^{(n+1)}}$.\label{prop:chainmap2}
    \end{enumerate}
\end{prop}

\begin{proof}
    Set, for simplicity, $S\coloneqq S^{(n)}$ and $T\coloneqq T^{(n+1)}=\Delta\circ_{r}S^{(n)}$. First, observe that since $c\in \operatorname{Col}_{T^{(n+1)}}(\Delta^{N+2})$, we have for any $(n-1)$-simplex $\tau\subset\Delta^{N+2}$
    \begin{equation*}
        T(c(\partial_1\tau), c(\partial_3\tau),\ldots, c(\partial_{2\lfloor\frac{n}{2}\rfloor-1}\tau)) = (c(\partial_0\tau), c(\partial_2\tau),\ldots, c(\partial_{2\lfloor\frac{n+1}{2}\rfloor-2}\tau)),
    \end{equation*}
    so that
    \begin{align}
        c(\partial_1\tau) &= c(\partial_0\tau),\label{eq:permit1}\\
        S(c(\partial_1\tau), c(\partial_3\tau),\ldots, c(\partial_{2\lfloor\frac{n}{2}\rfloor-1}\tau)) &= (c(\partial_2\tau), c(\partial_4\tau),\ldots, c(\partial_{2\lfloor\frac{n+1}{2}\rfloor-2}\tau)).\label{eq:permit2}
    \end{align}

    For \ref{prop:chainmap1}, let $\sigma=[v_0,\ldots,v_{n-2}]\subset\Delta^N$ be an $(n-2)$-simplex. Consider the $(n-1)$-simplex $\tau=[0, v_0+2,\ldots,v_{n-2}+2]\subset\Delta^{N+2}$.
    Now, by the definition of $\pi_{N+2}$, we have
    \begin{align*}
         (\pi_{N+2}c)(\partial_i\sigma) &=  (\pi_{N+2}c)([v_0,\ldots,\hat{v_i},\ldots,v_{n-2}])\\
         &= c([0,v_0+2,\ldots,\widehat{v_i+2},\ldots,v_{n-2}+2])\\
         &= c(\partial_{i+1}\tau),
    \end{align*}
    and so
    \begin{align*}
        \MoveEqLeft S( (\pi_{N+2}c)(\partial_0\sigma), (\pi_{N+2}c)(\partial_2\sigma),\ldots, (\pi_{N+2}c)(\partial_{2\lfloor\frac{n}{2}\rfloor-2}\sigma))\\
        &=S^{(n)}(c(\partial_1\tau), c(\partial_3\tau),\ldots, c(\partial_{2\lfloor\frac{n}{2}\rfloor-1}\tau))\\
        &= (c(\partial_2\tau), c(\partial_4\tau),\ldots, c(\partial_{2\lfloor\frac{n+1}{2}\rfloor-2}\tau)) & \text{(Eq.~\eqref{eq:permit2})}\\
        &= ( (\pi_{N+2}c)(\partial_1\sigma), (\pi_{N+2}c)(\partial_3\sigma),\ldots, (\pi_{N+2}c)(\partial_{2\lfloor\frac{n-1}{2}\rfloor-1}\sigma)).
    \end{align*}
    Hence, $\pi_{N+2}c\in \operatorname{Col}_{S^{(n)}}(\Delta^N)$.

    For \ref{prop:chainmap2}, we first show
    \begin{equation*}
        (d^S)^N_i\pi_{N+2} = \pi_{N+1} (d^T)^{N+2}_{i+2}\quad\text{for $0\leq i\leq N$}\quad\text{and}\quad \pi_{N+1} (d^T)^{N+2}_{0}=\pi_{N+1} (d^T)^{N+2}_{1}.
    \end{equation*}
    Let $c\in \operatorname{Col}_{T^{(n+1)}}(\Delta^{N+2})$, and $\sigma=[v_0,\ldots,v_{n-3}]\in\Delta^{N-1}_{n-3}$. We calculate
    \begin{align*}
        \MoveEqLeft ((d^S)^N_i\pi_{N+2}c)(\sigma) = (\pi_{N+2}c)(\delta^i\sigma)\\
        &= \begin{cases}
            c([0,v_0+2,\ldots,v_{n-3}+2]) & \text{if $i>v_{n-3}$,}\\
            c([0,v_0+2,\ldots,v_{k-1}+2,v_k+3,\ldots,v_{n-3}+3]) & \text{if $v_{k-1}<i\leq v_k$ for some $k$,}\\
            c([0,v_0+3,\ldots,v_{n-3}+3]) & \text{if $i\leq v_0$}
        \end{cases}
    \end{align*}
    for $0\le i\le N$, and
    \begin{align*}
        \MoveEqLeft (\pi_{N+1} (d^T)^{N+2}_{i+2}c)(\sigma) = c(\delta^{i+2}[0, v_0+2,\ldots,v_{n-3}+2])\\
        &= \begin{cases}
            c([0,v_0+2,\ldots,v_{n-3}+2]) & \text{if $i>v_{n-3}$,}\\
            c([0,v_0+2,\ldots,v_{k-1}+2,v_k+3,\ldots,v_{n-3}+3]) & \text{if $v_{k-1}<i\leq v_k$ for some $k$,}\\
            c([0,v_0+3,\ldots,v_{n-3}+3]) & \text{if $-1\leq i\leq v_0$,}\\
            c([1,v_0+3,\ldots,v_{n-3}+3]) & \text{if $i+2=0$}
        \end{cases}
    \end{align*}
    for $0\le i+2\le N$. Thus, $(d^S)^N_i\pi_{N+2} = \pi_{N+1} (d^T)^{N+2}_{i+2}$ for $0\leq i\leq N$. Also by Eq.~\eqref{eq:permit1} for $\tau=[0,1,v_0+3,\ldots,v_{n-3}+3]$, we obtain $\pi_{N+1} (d^T)^{N+2}_{0}=\pi_{N+1} (d^T)^{N+2}_{1}$.

    Now the proof proceeds as:
    \begin{align*}
        \pi_{N+1}\partial^T_{N+2} &= \pi_{N+1}(d^T)^{N+2}_{0}-\pi_{N+1} (d^T)^{N+2}_{1} + \sum_{i=2}^{N+2}(-1)^i\pi_{N+1}(d^T)^{N+2}_{i}\\
        &= \sum_{i=0}^{N}(-1)^i(d^S)^N_{i}\pi_{N+2} = \partial^S_N\pi_{N+2}.
    \end{align*}
\end{proof}

\begin{ex}[Classifying spaces of groups]
    Let $X$ be a group and consider, as before, the solutions $T=T^{(5)}\colon X\times X\to X\times X$ and $S=S^{(4)}\colon X\times X\to X$ of the pentagon and dual 4-gon equations given by
    \begin{equation*}
        T(x,y)\coloneqq (x,xy) \quad\text{and}\quad S(x,y)\coloneqq xy,
    \end{equation*}
    respectively. We saw in Section~\ref{subsec:cohomologyforsmalln} that the semi-simplicial sets $\operatorname{Col}_T$ and $\operatorname{Col}_S$ correspond to the two-sided bar construction $B(*,X,X)$ and $B(*,X,*)$, respectively. In this view, the map $\pi_{N+2}$ in Proposition~\ref{prop:chainmap} corresponds to the projection $(x_1,\ldots,x_{N+1})\mapsto (x_1,\ldots, x_N)$. 

    More generally, for the solutions $T^{(2k+1)}$ and $S^{(2k)}$ given by Eq.~(\ref{eq:Tcommutativemonoid}) and (\ref{eq:Scommutativemonoid}), respectively, this example, together with Corollary~\ref{cor:homologymonoid}, suggests that the chain map $\{\pi_{N+2}\colon \operatorname{Col}_{T^{(2k+1)}}(\Delta^{N+2})\to \operatorname{Col}_{S^{(2k)}}(\Delta^N)\}$ may be interpreted as a higher-dimensional analog of the above usual case.
\end{ex}

\section{Polygon equations and the higher Segal condition}\label{sec:segal}

The semi-simplicial set $\operatorname{Col}_T$ associated in Definition \ref{def:polygon_coloring} with a solution of a (dual) polygon equation satisfies an additional property called the higher Segal condition.

\subsection{Simplicial vs. semi-simplicial}

We first show that, in general, the semi-simplicial set $\operatorname{Col}_T$ does not extend to a simplicial set. This explains why we work in the framework of semi-simplicial sets.

\begin{prop}
    Let $T$ be a solution of the $n$-gon equation for $n\geq 4$ or a solution of the dual $n$-gon equation for $n\geq 5$ over a set $X$.
    Then $\operatorname{Col}_T$ admits degeneracy maps extending its semi-simplicial structure to a simplicial structure if and only if $X$ is a singleton.
\end{prop}

\begin{proof}
    Let $T$ be a solution of the (dual) $n$-gon equation over $X$.
    Assume $\operatorname{Col}_T$ is a simplicial set with the underlying semi-simplicial structure given in Section~\ref{subsec:poly_cohomology} and denote the degeneracy maps by $\{s_i^{m}\colon (\operatorname{Col}_T)_m\to (\operatorname{Col}_T)_{m+1}\}_{0 \leq i \leq m}$.
    Since $(\operatorname{Col}_T)_{n-4}=\{*\}$ and $(\operatorname{Col}_T)_{n-3}=X$, we set $x_i \coloneqq s_i^{n-4}(*)\in X$ for $0\leq i \leq n-4 $.
    From the compatibility condition of $d$ and $s$, for any $x\in X$ we have
    \begin{equation*}
        \begin{cases}
            d_is_j=s_{j-1}d_i & \text{if} \quad i < j\\
            d_is_j=\mathrm{id} & \text{if} \quad i = j,j+1\\
            d_is_j=s_{j}d_{i-1} & \text{if} \quad i > j+1
        \end{cases}
        \quad
        \Longrightarrow
        \quad
        d_i^{n-2}(s_j^{n-3}(x)) = 
        \begin{cases}
            x_{j-1} & \text{if} \quad i < j\\
            x & \text{if}\quad i = j, j+1\\
            x_{j} & \text{if}\quad i > j+1\\
        \end{cases}
    \end{equation*}
    From the relation $s_i^{n-3}s_i^{n-4}=s_{i+1}^{n-3}s_i^{n-4}$, we have 
    \begin{align*}
        x_{i}&=d^{n-2}_{n-2}s_i^{n-3}s_i^{n-4}(*)=d^{n-2}_{n-2}s_{i+1}^{n-3}s_i^{n-4}(*)=x_{i+1}, \qquad 0\leq i <n-4. 
    \end{align*}
    We denote the unique element by $x_*(=x_i \text{ for any $i$})$.
    
    Now, assume $T$ is a solution of the $n$-gon equation ($n\geq 4$).
    From the definition of $T$-colorings, $c, c^{\prime}\in (\operatorname{Col}_T)_{n-2}$ coincide if and only if $d^{n-2}_i(c)=d^{n-2}_i(c^{\prime})$ for all odd $i$.
    Since $d_i^{n-2}s_0^{n-3}(x)=d_i^{n-2}s_1^{n-3}(x)$ for all odd $i$, the two maps $s_0^{n-3}$ and $s_1^{n-3}$ coincide.
    Thus, 
    \[
    x = d_0^{n-2}s_0^{n-3}(x)=d_0^{n-2}s_1^{n-3}(x)=x_*
    \]
    for any $x\in X$, and $X$ is a singleton.

    Similarly, assume $T$ is a solution of the dual $n$-gon equation ($n\geq 5$).
    From the definition of $T$-colorings, $c, c^{\prime}\in (\operatorname{Col}_T)_{n-2}$ coincide if and only if $d^{n-2}_i(c)=d^{n-2}_i(c^{\prime})$ for all even $i$.
    Since $d_i^{n-2}s_{1}^{n-3}(x)=d_i^{n-2}s_{2}^{n-3}(x)$ for all even $i$, the two maps $s_{1}^{n-3}$ and $s_{2}^{n-3}$ coincide.
    Thus, 
    \[
    x = d_{1}^{n-2}s_{1}^{n-3}(x)=d_{1}^{n-2}s_{2}^{n-3}(x)=x_*
    \]
    for any $x\in X$, and $X$ is a singleton.
    
    Conversely, if X is a singleton, then every $(\operatorname{Col}_T)_m$ is a singleton, and the unique possible degeneracy maps make $\operatorname{Col}_T$ into the terminal simplicial set.
\end{proof}

\subsection{Relation to higher Segal semi-simplicial sets}

We investigate the relationship between polygon equations and higher Segal semi-simplicial sets. 

For any collection $\mathcal{I}$ of simplices of $\Delta^N=[0,1,\ldots,N]$, let $\langle\mathcal{I}\rangle$ denote the subcomplex generated by $\mathcal{I}$, that is,
\begin{equation*}
    \langle \mathcal{I} \rangle \coloneqq \{\,\tau\subset\Delta^N\mid \tau\subset\sigma \text{ for some } \sigma\in\mathcal{I}\,\}.
\end{equation*}
By convention, we remove the empty face from $\langle\mathcal{I}\rangle$, and regard the resulting set $\langle\mathcal{I}\rangle^{\times}$ as a poset (hence a small category) under inclusion.

For $N\geq d\geq 0$, define two posets $ \mathcal{U}(N,d)$ and  $\mathcal{L}(N,d)$ in $\Delta^N$ as
\begin{equation*}
    \mathcal{U}(N,d) \coloneqq \begin{cases}
        \langle \mathcal{P}_d^N\rangle^{\times} & \text{if $d$ is even,}\\
        \langle\mathcal{N}_d^N\rangle^{\times} & \text{if $d$ is odd,}
    \end{cases} \quad\text{and}\quad \mathcal{L}(N,d) \coloneqq \begin{cases}
        \langle \mathcal{N}_d^N\rangle^{\times} & \text{if $d$ is even,}\\
        \langle\mathcal{P}_d^N\rangle^{\times} & \text{if $d$ is odd,}
    \end{cases}
\end{equation*}
respectively. Note that these posets coincide with the upper and lower triangulations of the cyclic polytope.

Given a semi-simplicial set $Z$ and an $n$-simplex $\sigma \in \langle \mathcal{I}\rangle^{\times}$, we denote by $Z_{\sigma}$ a copy of $Z_n$.
Furthermore, for a face $\tau \subset \sigma$, there exists a canonical map $Z_{\sigma}\to Z_{\tau}$ induced by the inclusion. This forms a functor $(\langle \mathcal{I}\rangle^{\times})^{\operatorname{op}} \to \mathsf{Set}$.

\begin{defi}[\cite{dyckerhoff2026cyclic,poguntke2017higher}] 
    Let $n\geq 0$. A semi-simplicial set $Z$ is
    \begin{itemize}
        \item \textbf{upper} $\boldsymbol{n}$\textbf{-Segal} if the natural map
        \begin{equation*}
            Z_N \longrightarrow \lim_{\sigma\in\mathcal{U}(N,n)^\mathrm{op}} Z_\sigma
        \end{equation*}
        is bijective for all $N \geq n$,
        \item \textbf{lower} $\boldsymbol{n}$\text{-Segal} if the natural map
        \begin{equation*}
            Z_N \longrightarrow \lim_{\sigma\in\mathcal{L}(N,n)^\mathrm{op}} Z_\sigma
        \end{equation*}
        is bijective for all $N \geq n$,
        \item \textbf{fully} $\boldsymbol{n}$\textbf{-Segal} if it is both upper and lower $n$-Segal.
    \end{itemize}
    
\end{defi}

Let $n\text{-}\mathsf{Gon}$ denote the category of set-theoretic solutions of $n$-gon equations. Its objects are pairs $(X,T)$, where $T$ is a solution of the $n$-gon equation over a set $X$, and its morphisms $(X,T) \to (Y,U)$ are maps $\alpha\colon X \to Y$ satisfying $\alpha^{\times \lfloor \frac{n}{2} \rfloor}\circ T = U \circ \alpha^{\times \lfloor \frac{n-1}{2} \rfloor}$. Similarly, let $n\text{-}\mathsf{Gon}^*$ denote the category of set-theoretic solutions of the dual $n$-gon equation.

Let $n\text{-}\mathsf{Segal}^{<m}_{\mathcal{L}}$, $n\text{-}\mathsf{Segal}^{<m}_{\mathcal{U}}$ and $n\text{-}\mathsf{Segal}^{<m}$ denote, respectively, the categories of lower, upper, and fully $n$-Segal semi-simplicial sets satisfying $Z_N = \{*\}$ for $N < m$. The morphisms of these categories are the usual morphisms of semi-simplicial sets. 

The following proposition is a generalization of \cite[Corollary 3.7.4]{dyckerhoff2019higher}.

\begin{thm}\label{thm:polygontoSegal}
    We have the following:
    \begin{enumerate}[label=(\arabic*)]
        \item The categories $(2k+1)\text{-}\mathsf{Gon}$ and $(2k-2)\text{-}\mathsf{Segal}^{<2k-2}_{\mathcal{L}}$ are equivalent. \label{thm:polygontoSegal1}
        \item The categories $(2k+1)\text{-}\mathsf{Gon}^*$ and $(2k-2)\text{-}\mathsf{Segal}^{<2k-2}_{\mathcal{U}}$ are equivalent. \label{thm:polygontoSegal2}
        \item Let $(2k+1)\text{-}\mathsf{Gon}^{\mathrm{bij}}$ be the full subcategory of $(2k+1)\text{-}\mathsf{Gon}$ consisting of all bijective solutions of the $(2k+1)$-gon equation. Then the categories $(2k+1)\text{-}\mathsf{Gon}^{\mathrm{bij}}$ and $(2k-2)\text{-}\mathsf{Segal}^{<2k-2}$ are equivalent. \label{thm:polygontoSegal3}
        \item The categories $2k\text{-}\mathsf{Gon}$ and $(2k-3)\text{-}\mathsf{Segal}^{<2k-3}_{\mathcal{U}}$ are equivalent. \label{thm:polygontoSegal4}
        \item The categories ${2k}\text{-}\mathsf{Gon}^*$ and ${(2k-3)}\text{-}\mathsf{Segal}^{<2k-3}_{\mathcal{L}}$ are equivalent. \label{thm:polygontoSegal5}
    \end{enumerate}
\end{thm}
\begin{proof}
    Let us prove \ref{thm:polygontoSegal1}. Part \ref{thm:polygontoSegal2}, \ref{thm:polygontoSegal4}, and \ref{thm:polygontoSegal5} are similar.

    Let us define 
    \[
    \mathcal{F}\colon{(2k+1)}\text{-}\mathsf{Gon} \longrightarrow {(2k-2)}\text{-}\mathsf{Segal}^{<2k-2}_{\mathcal{L}}
    \]
    by sending the pair $(X,T)$ to the semi-simplicial set $\operatorname{Col}_T$.
    We need to show that the canonical map
    \[
    (\operatorname{Col}_T)_N\longrightarrow \lim_{\sigma\in\mathcal{L}(N,2k-2)^\mathrm{op}} (\operatorname{Col}_T)_\sigma
    \]
    is bijective for all $N<2k-2$.
    Since $(\operatorname{Col}_T)_N=\{*\}$ for $N<2k-2$, the limit $\lim_{\sigma\in\mathcal{L}(N,2k-2)^\mathrm{op}} (\operatorname{Col}_T)_\sigma$ becomes a product $\prod_{\sigma\in\mathcal{N}^N_{2k-2}}\!\!\!\operatorname{Col}_T(\Delta^{2k-2})$, and the canonical map is given by restriction
    \begin{align*}
        &\operatorname{Col}_T(\Delta^N) \ni c \longmapsto \prod_{\sigma\in\mathcal{N}^N_{2k-2}}\!\!\!c|_{\sigma} \in\prod_{\sigma\in\mathcal{N}^N_{2k-2}}\!\!\!\operatorname{Col}_T(\Delta^{2k-2})
    \end{align*}
    By Proposition~\ref{prop:coloring_of_absolutely_positive/negative}, this map is bijective, thus $\operatorname{Col}_T$ is a lower $(2k-2)$-Segal semi-simplicial set.

    Let $\alpha$ be morphism from $(X,T)$ and $(Y,U)$, i.e., $\alpha\colon X \to Y$ such that $\alpha^{\times k}\circ T=U \circ \alpha^{\times k}$.
    Then we define $\mathcal{F}(\alpha)$ by $\operatorname{Col}_T(\Delta^N)\ni c \mapsto \alpha\circ c \in \operatorname{Col}_U(\Delta^N)$ for each $N$. The fact that $\alpha\circ c$ is a $U$-coloring follows from the condition $\alpha$ satisfies.

    Conversely, let us define 
    \[
    \mathcal{G}\colon{(2k-2)}\text{-}\mathsf{Segal}^{<2k-2}_{\mathcal{L}} \longrightarrow {(2k+1)}\text{-}\mathsf{Gon}
    \]
    as follows.
    Given a lower $(2k-2)$-Segal semi-simplicial set $Z\in{(2k-2)}\text{-}\mathsf{Segal}^{<2k-2}_{\mathcal{L}}$, we have two canonical maps $\phi$ and $\psi$:

    \begin{equation}\label{dig:defTz}
        \begin{tikzcd}[column sep=large, row sep=large]
        {\displaystyle{\lim_{\sigma\in\mathcal{L}(2k-1,2k-2)^{\mathrm{op}}}Z_{\sigma}}}
        &&
        {\displaystyle{\lim_{\sigma\in\mathcal{U}(2k-1,2k-2)^{\mathrm{op}}}Z_{\sigma}}}
        \\
        & {Z_{2k-1}}
        \arrow[ur, "\psi"']
        \arrow[ul, "\phi"]
        \arrow[dashrightarrow, from=1-1, to=1-3, "T_Z"]
        \end{tikzcd}
    \end{equation}
    Now, since $Z_N=\{*\}$ for $N<2k-2$, we have
    \begin{align*}
        \lim_{\sigma\in\mathcal{L}(2k-1,2k-2)^{\mathrm{op}}}Z_{\sigma}=\prod_{\sigma\in\mathcal{N}^{2k-1}_{2k-2}}\!\!\!Z_\sigma=(Z_{2k-2})^{\times k}\quad\text{and}\quad
        \lim_{\sigma\in\mathcal{U}(2k-1,2k-2)^{\mathrm{op}}}Z_{\sigma}=\prod_{\sigma\in\mathcal{P}^{2k-1}_{2k-2}}\!\!\!Z_\sigma=(Z_{2k-2})^{\times k}.
    \end{align*}
    By the lower $(2k-2)$-Segal condition, $\phi$ is a bijection, and thus we set $T_Z\coloneqq\psi\circ\phi^{-1}\colon (Z_{2k-2})^{\times k} \to (Z_{2k-2})^{\times k}$.
    We show that this $T_Z \colon (Z_{2k-2})^{\times k} \to (Z_{2k-2})^{\times k}$ satisfies the $(2k+1)$-gon equation.
    
    Recall that in Section \ref{sec:polygon_equations}, we considered the partision
    \begin{equation*}
        D^{+} = \{\,\partial_i\Delta^{2k}\,|\,\text{$0\leq i \leq 2k$, $i$ is even}\,\},\qquad
        D^{-} = \{\,\partial_j\Delta^{2k}\,|\,\text{$0\leq j \leq 2k$, $j$ is odd}\,\}
    \end{equation*}
    of $(2k-1)$-dimensional faces of $\Delta^{2k}$.
    Consider the following commutative diagram:
        \begin{equation}\label{dig:2k+1_gon_diagram}
        \begin{tikzcd}[row sep=large, column sep=large]
        &
        {\displaystyle{\lim_{\sigma\in(\langle D^+\rangle^{\times})^{\mathrm{op}}}} Z_{\sigma}}
        \\
        {\displaystyle{
        \prod_{\sigma\in\mathcal{N}^{2k}_{2k-2}}\!\!\!Z_\sigma
        =
        \lim_{\sigma\in\mathcal{L}(2k,2k-2)^{\mathrm{op}}}}Z_{\sigma}}
        &
        {Z_{2k}}
        &
        {\displaystyle{
        \lim_{\sigma\in\mathcal{U}(2k,2k-2)^{\mathrm{op}}}}Z_{\sigma}
        =
        \prod_{\sigma\in\mathcal{P}^{2k}_{2k-2}}\!\!\!Z_\sigma}
        \\
        &
        {\displaystyle{\lim_{\sigma\in(\langle D^-\rangle^{\times})^{\mathrm{op}}}} Z_{\sigma}}
        \arrow["", from=2-2, to=1-2]
        \arrow[""', from=2-2, to=2-1]
        \arrow["", from=2-2, to=2-3]
        \arrow[""', from=2-2, to=3-2]
        \arrow["r_l"', from=1-2, to=2-1]
        \arrow["r_u", from=1-2, to=2-3]
        \arrow["s_l", from=3-2, to=2-1]
        \arrow["s_u"', from=3-2, to=2-3]
        \end{tikzcd}
        \end{equation}

    We prove that $r_l$ is bijective by constructing an inverse map $\overline{r}_l$.
    By the universal property of limits, to construct a map 
    \begin{equation*}
    \overline{r}_l\colon\prod_{\sigma\in\mathcal{N}^{2k}_{2k-2}}\!\!\!Z_\sigma \longrightarrow \lim_{\sigma\in(\langle D^+\rangle ^{\times})^{\mathrm{op}}}Z_{\sigma},
    \end{equation*}
    we just need to define a map 
    \begin{equation*}
    (\overline{r}_l)_{\tau}\colon \prod_{\sigma\in\mathcal{N}^{2k}_{2k-2}}\!\!\!Z_\sigma \longrightarrow Z_{\tau}
    \end{equation*}
    for each $\tau \in D^+$ such that the following diagram commutes whenever $\tau$ and $\tau^{\prime}$ share the common $(2k-2)$-simplex $\tau\cap \tau^{\prime}$:
    \begin{equation}\label{dig:overlinercommute}
        \begin{tikzcd}
        {\displaystyle{\prod_{\sigma\in\mathcal{N}^{2k}_{2k-2}}\!\!\!Z_\sigma}} & {Z_{\tau^{\prime}}}\\
        {Z_{\tau}} & {Z_{\tau\cap\tau^{\prime}}}
        \arrow[" (\overline{r}_l)_{\tau}"',from=1-1, to=2-1]
        \arrow[" (\overline{r}_l)_{\tau^{\prime}}", from=1-1, to=1-2]
        \arrow[""',from=2-1, to=2-2]
        \arrow[""',from=1-2, to=2-2]
    \end{tikzcd}
    \end{equation}
    We define $(\overline{r}_l)_{\partial_{2k-2i}\Delta^{2k}}$ ($i=0,\ldots,k$) recursively as follows.
    First, for $i=0$, since the negative faces of $\partial_{2k}\Delta^{2k}$ belong to $\mathcal{N}^{2k}_{2k-2}$, we define $(\overline{r}_l)_{\partial_{2k}\Delta^{2k}}$ as the composition
    \[
    \prod_{\sigma\in\mathcal{N}^{2k}_{2k-2}}\!\!\!Z_\sigma \xlongrightarrow{\pi} \prod_{\substack{\text{$\sigma$: negative}\\ \text{w.r.t. $\partial_{2k}\Delta^{2k}$}}}\!\!\!Z_\sigma\cong\prod_{\sigma\in\mathcal{N}^{2k-1}_{2k-2}}\!\!\!Z_\sigma \xlongrightarrow{\phi^{-1}} Z_{\partial_{2k}\Delta^{2k}},
    \]
    where $\pi$ is the projection.

    Next, define $\mathcal{N}^{2k}_{2k-2}(i)$ for $0\leq i \leq k+1$ recursively as follows:
    \begin{align*}
        \mathcal{N}^{2k}_{2k-2}(0) &\coloneqq \mathcal{N}^{2k}_{2k-2},\\
        \mathcal{N}^{2k}_{2k-2}(i) &\coloneqq (\mathcal{N}^{2k}_{2k-2}(i-1)\setminus \{\text{negative faces of $\partial_{2k-2i+2}\Delta^{2k}$}\})\cup\{\text{positive faces of $\partial_{2k-2i+2}\Delta^{2k}$}\}.
    \end{align*}
    We also define the map $T_Z(i)$ for $0\leq i\leq k$ as:
    \begin{equation}
    \begin{tikzcd}
        {\displaystyle{\prod_{\sigma\in\mathcal{N}^{2k}_{2k-2}(i)}\!\!\!Z_\sigma}} && {\displaystyle{\prod_{\sigma\in\mathcal{N}^{2k}_{2k-2}(i+1)}\!\!\!Z_\sigma}}\\
        {\displaystyle{\prod_{\substack{{\sigma\in\mathcal{N}^{2k}_{2k-2}(i)}\\{\sigma \notin \partial_{2k-2i}\Delta^{2k}}}}\!\!\!Z_\sigma}\times \prod_{\substack{\text{$\sigma$: negative}\\ \text{w.r.t. $\partial_{2k}\Delta^{2k}$}}}\!\!\!Z_\sigma} && {\displaystyle{\prod_{\substack{{\sigma\in\mathcal{N}^{2k}_{2k-2}(i)}\\{\sigma \notin \partial_{2k-2i}\Delta^{2k}}}}\!\!\!Z_\sigma}\times \prod_{\substack{\text{$\sigma$: positive}\\ \text{w.r.t. $\partial_{2k}\Delta^{2k}$}}}\!\!\!Z_\sigma}
        \arrow[equal, ""',from=1-1, to=2-1]
        \arrow[dashrightarrow, " T_Z(i)", from=1-1, to=1-3]
        \arrow["\text{id}\times T_Z"',from=2-1, to=2-3]
        \arrow[equal, ""',from=1-3, to=2-3]
    \end{tikzcd}
    \end{equation}
    Then for $i>0$, define $(\overline{r}_l)_{\partial_{2k-2i} \Delta^{2k}}$ as the composition
    \begin{equation*}
        \begin{aligned}
            \prod_{\sigma\in\mathcal{N}^{2k}_{2k-2}(0)}\!\!\! Z_\sigma
            \xlongrightarrow{T_Z(0)}
            &\prod_{\sigma\in\mathcal{N}^{2k}_{2k-2}(1)}\!\!\! Z_\sigma
            \xlongrightarrow{T_Z(1)}\cdots
            \xlongrightarrow{T_Z(i-1)}
            \prod_{\sigma\in\mathcal{N}^{2k}_{2k-2}(i)}\!\!\! Z_\sigma\\
            &\xlongrightarrow{\pi}
            \prod_{\substack{\text{$\sigma$: negative}\\ \text{w.r.t. $\partial_{2k-2i}\Delta^{2k}$}}}\!\!\! Z_\sigma
            \cong
            \prod_{\sigma\in\mathcal{N}^{2k-1}_{2k-2}}\!\!\! Z_\sigma
            \xlongrightarrow{\phi^{-1}}
            Z_{\partial_{2k}\Delta^{2k}} .
        \end{aligned}
    \end{equation*}
    
    We have defined maps $(\overline{r}_l)_{\tau}\colon \prod_{\sigma\in\mathcal{N}^{2k}_{2k-2}}\!\!\!Z_\sigma \longrightarrow Z_{\tau}$ for $\tau\in D^+$, and now commutativity~(\ref{dig:overlinercommute}) is immediate from the construction. Hence, we obtain the map
    \begin{equation*}
        \overline{r}_l\colon\prod_{\sigma\in\mathcal{N}^{2k}_{2k-2}}\!\!\!Z_\sigma \longrightarrow \lim_{\sigma\in(\langle D^+\rangle ^{\times})^{\mathrm{op}}}Z_{\sigma}.
    \end{equation*}
    It remains to show that $r_l$ and $\overline{r}_l$ are inverse to each other. On the one hand, by the construction of $\overline{r}_l$, for each $\tau\in\mathcal{N}^{2k}_{2k-2}$, its $\tau$-component $(\overline{r}_l)_\tau$ is the projection onto the corresponding factor $Z_\tau$. Hence, $r_l\circ\overline{r}_l=\mathrm{id}$. On the other hand, again by the construction of $\overline{r}_l$, the $\partial_{2k-2i}\Delta^{2k}$-component of $\overline{r}_l\circ r_l$ ($i=0,\ldots,k$) is
    \begin{equation*}
        (\overline{r}_l\circ r_l)_{\partial_{2k-2i}\Delta^{2k}} = \phi^{-1}\circ \pi\circ T_z(i-1)\circ\cdots\circ T_z(0)\circ r_l = \pi_{\partial_{2k-2i}\Delta^{2k}},
    \end{equation*}
    where $\pi$ is the projection appearing in the construction of $\overline{r}_l$, and $\pi_{\partial_{2k-2i}\Delta^{2k}}\colon\lim_{\sigma\in(\langle D^+\rangle ^{\times})^{\mathrm{op}}}Z_{\sigma}\to Z_{\partial_{2k-2i}\Delta^{2k}}$ is the canonical projection. Thus, $\overline{r}_l\circ r_l=\mathrm{id}$ and $r_l$ is bijective.
     
    Observe that by the construction of $\overline{r}_l$, the composition $r_u\circ \overline{r}_l$ coincides precisely with the left-hand side of the $(2k+1)$-gon equation.

    Similarly, we can show that $s_u$ is bijective, with inverse $\overline{s}_l$ constructed similarly to $\overline{r}_u$.
    In this case, $s_u\circ \overline{s}_l$ coincides with the right-hand side of the $(2k+1)$-gon equation.
    
    Note that by the $(2k-2)$-Segal condition, the two horizontal maps out of $Z_{2k}$ in diagram \eqref{dig:2k+1_gon_diagram} are bijective. Therefore, by the 2-out-of-3 property, the two vertical maps are also bijective. The commutativity of diagram \eqref{dig:2k+1_gon_diagram} gives
    \begin{equation*}
        r_u\circ \overline{r}_l=s_u\circ \overline{s}_l,
    \end{equation*}
    which shows $T_z$ is a solution of the $(2k+1)$-gon equation.  

    To verify naturality, let $\beta \colon Z \to W$ be a morphism in $(2k-2)$-$\operatorname{Segal}^{<2k-2}_{\mathcal{L}}$.
    If we set $\mathcal{G}(\beta)\coloneqq \beta_{2k-2}\colon Z_{2k-2}\to W_{2k-2}$, then $\beta_{2k-2}^{\times k}\circ T_Z = T_W \circ \beta_{2k-2}^{\times k}$. Hence, $\mathcal{G}(\beta)$ defines a morphism from $T_Z$ to $T_W$.

    Finally, it remains to check that $\mathcal{G}\circ\mathcal{F}=\operatorname{id}$ and $\mathcal{F}\circ\mathcal{G}=\operatorname{id}$, which is immediate from the definitions of $\mathcal{F}$ and $\mathcal{G}$.

    Part~\ref{thm:polygontoSegal3} follows from part~\ref{thm:polygontoSegal1}, \ref{thm:polygontoSegal2} and Proposition~\ref{prop:cohomologybijective}.
\end{proof}

\begin{ex}
    Here, we look closely at the above proof in the case of $2$-Segal semi-simplicial sets. 
    
    Let $Z$ be a lower $2$-Segal semi-simplicial set with $Z_0=Z_1=\{*\}$. Then diagram~(\ref{dig:defTz}) becomes
    \begin{equation*}
        \begin{tikzcd}[column sep=large, row sep=large]
        {Z_{023}\times Z_{012}=Z_2\times Z_2}
        &&
        {Z_2\times Z_2=Z_{123}\times Z_{013}}
        \\
        & {Z_3=Z_{0123}}
        \arrow[ur, "{\psi=(d_0,d_2)}"']
        \arrow[ul, "{\phi=(d_1,d_3)}", "\cong"'{sloped, above}]
        \arrow[dashrightarrow, from=1-1, to=1-3, "T_Z=T_Z(0123)"]
        \end{tikzcd}
    \end{equation*}
    Hence, the map $T=T_Z\colon Z_2\times Z_2\to Z_2\times Z_2$ is given elementwise by
    \begin{equation*}
        (d_1x,d_3x) \longmapsto (d_0x,d_2x)\quad x\in Z_3
    \end{equation*}
    Further,
    \begin{equation*}
        \lim_{\sigma\in\mathcal{L}(2k,2k-2)^{\mathrm{op}}}Z_{\sigma} = 
        \prod_{\sigma\in\mathcal{N}^{2k}_{2k-2}}\!\!\!Z_\sigma
        = Z_{034}\times Z_{023}\times Z_{012},
    \end{equation*}
    and ${\lim_{\sigma\in(\langle D^+\rangle^{\times})^{\mathrm{op}}}} Z_{\sigma}$ is the limit of the following diagram (compare this with Figure~\ref{fig:graphical45goneq}):
    \begin{equation*}
        \includegraphics[]{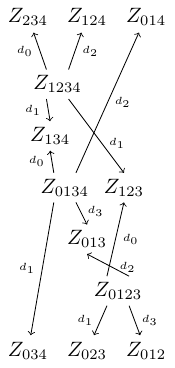}
    \end{equation*}
   So, for example, the $[134]$-component of the map $\overline{r}_l$ is given by
    \begin{equation*}
        Z_{034}\times Z_{023}\times Z_{012}\xrightarrow{T_Z(0)} Z_{034}\times Z_{123}\times Z_{013}\xlongrightarrow{\pi} Z_{034}\times Z_{013}\xlongrightarrow{\phi^{-1}} Z_{0134}.
    \end{equation*}
    Finally, by identifying $Z_4$ with $Z_2^{\times 3}$ via
    \begin{equation*}
        y\longmapsto (d_1d_2y,d_1d_4y,d_3d_4y)\quad y\in Z_4,
    \end{equation*}
    (which corresponds to the leftward horizontal map out of $Z_{2k}$ in diagram~(\ref{dig:2k+1_gon_diagram})) one can see that both maps $r_u\circ \overline{r}_l=T_{12}T_{13}T_{23}$ and $s_u\circ \overline{s}_l=T_{23}T_{12}$ in the pentagon equation send $(d_1d_2y,d_1d_4y,d_3d_4y)$ to $(d_0d_1y,d_0d_3y,d_2d_3y)$ (using $d_id_j=d_{j-1}d_i$ for $i<j$).
\end{ex}

\begin{rem}
    Given a set-theoretic solution of the $n$-simplex equation, Korepanov--Sharygin--Talalaev constructed in \cite{korepanov2016cohomologies} a set of permitted colorings of the $N$-cube $I^{N}$, thus forming a semi-cubical set.
    As in our case, it is natural to ask whether or not this cubical set satisfies some kind of higher Segal condition.
    In order to formulate this question, one first needs an appropriate definition of the higher Segal condition of (semi-)cubical sets.
    For example, a potential candidate could be given by considering limits over canonical cubillage of a cyclic zonotope, corresponding to the minimal or maximal element of the higher Bruhat order~\cite{kapranov1991combinatorial} (see also, for example, \cite{williams2023first}).
\end{rem}

\bibliography{references}

@article{dimakis2021grassmannian,
  title={Grassmannian-parameterized solutions to direct-sum polygon and simplex equations},
  author={Dimakis, Aristophanes and Korepanov, Igor G},
  journal ={J. Math. Phys.},
  fjournal={Journal of Mathematical Physics},
  volume={62},
  number={5},
  year={2021},
  publisher={AIP Publishing}
}

@article{dimakis2015simplex,
  title={Simplex and polygon equations},
  author={Dimakis, Aristophanes and M{\"u}ller-Hoissen, Folkert},
  journal={SIGMA Symmetry Integrability Geom. Methods Appl.},
  fjournal={SIGMA. Symmetry, Integrability and Geometry: Methods and Applications},
  volume={11},
  pages={042, 49},
  year={2015},
  publisher={SIGMA. Symmetry, Integrability and Geometry: Methods and Applications}
}

@article{kashaev1998pentagon,
  title={On pentagon, ten-term, and tetrahedron relations},
  author={Kashaev, Rinat M and Sergeev, Sergey M},
  journal = {Comm. Math. Phys.},
  fjournal={Communications in Mathematical Physics},
  volume={195},
  number={2},
  pages={309--319},
  year={1998},
  publisher={Springer}
}

@article{muller2024structure,
  title={On the structure of set-theoretic polygon equations},
  author={M\"uller-Hoissen, Folkert},
  journal={SIGMA Symmetry Integrability Geom. Methods Appl.},
  fjournal={SIGMA. Symmetry, Integrability and Geometry. Methods and Applications},
  volume={20},
  pages={Paper No. 051, 30},
  year={2024},
}

@article{kashaev2015realizations,
  title={On realizations of {P}achner moves in 4d},
  author={Kashaev, Rinat M},
  journal= {J. Knot Theory Ramifications},
  fjournal={Journal of Knot Theory and Its Ramifications},
  volume={24},
  number={13},
  pages={1541002},
  year={2015},
  publisher={World Scientific}
}

@article{korepanov2024odd,
  title={Odd-gon relations and their cohomology},
  author={Korepanov, Igor G},
  journal={Partial Differential Equations in Applied Mathematics},
  volume={11},
  pages={100856},
  year={2024},
  publisher={Elsevier}
}

@article{bazhanov1982conditions,
  title={Conditions of commutativity of transfer matrices on a multidimensional lattice},
  author={Bazhanov, Vladimir V and Stroganov, Yurii Grigor'evich},
  fjournal={Akademiya Nauk SSSR. Teoreticheskaya i Matematicheskaya Fizika},
  journal={Teoret. Mat. Fiz.},
  volume={52},
  number={1},
  pages={105--113},
  year={1982},
  publisher={Russian Academy of Sciences, Steklov Mathematical Institute of Russian~…}
}

@article{zamolodchikov1980tetrahedra,
  title={Tetrahedra equations and integrable systems in three-dimensional space},
  author={Zamolodchikov, AB},
  journal={Sov. Phys. JETP},
  fjournal={Soviet Physics. JETP},
  volume={52},
  number={2},
  pages={325--336},
  year={1980}
}

@article{zamolodchikov1981tetrahedron,
  title={Tetrahedron equations and the relativistic {$S$}-matrix of straight-strings in {$2+1$}-dimensions},
  author={Zamolodchikov, AB},
  journal={Comm. Math. Phys.},
  fjournal={Communications in Mathematical Physics},
  volume={79},
  number={4},
  pages={489--505},
  year={1981},
  publisher={Springer}
}

@misc{korepanov2017hexagon,
  title={Hexagon cohomologies and polynomial {TQFT} actions},
  author={Korepanov, Igor G and Sadykov, Nurlan M},
  note={preprint. arXiv:1707.02847]},
  year={2017}
}

@article{korepanov2016cohomologies,
  title={Cohomologies of {$n$}-simplex relations},
  author={Korepanov, Igor G and Sharygin, Georgy I and Talalaev, Dmitry V},
  journal= {Math. Proc. Cambridge Philos. Soc.},
  fjournal={Mathematical Proceedings of the Cambridge Philosophical Society},
  volume={161},
  number={2},
  pages={203--222},
  year={2016},
  organization={Cambridge University Press}
}

@article{yang1967some,
  title={Some exact results for the many-body problem in one dimension with repulsive delta-function interaction},
  author={Yang, Chen-Ning},
  journal= {Phys. Rev. Lett.},
  fjournal={Physical Review Letters},
  volume={19},
  pages={1312--1315},
  year={1967},
  publisher={APS},
  issue={23},
}

@article{baxter1972partition,
  title={Partition function of the eight-vertex lattice model},
  author={Baxter, Rodney J},
  journal= {Ann. Physics},
  fjournal={Annals of Physics},
  volume={70},
  number={1},
  pages={193--228},
  year={1972},
  publisher={Elsevier}
}

@article{wan2024matrix,
    AUTHOR = {Wan, Zheyan},
     TITLE = {A matrix solution to any polygon equation},
   JOURNAL = {Adv. Theor. Math. Phys.},
  FJOURNAL = {Advances in Theoretical and Mathematical Physics},
    VOLUME = {29},
      YEAR = {2025},
    NUMBER = {7},
     PAGES = {1827--1855},
      ISSN = {1095-0761,1095-0753},
   MRCLASS = {15A24},
  MRNUMBER = {4992687},
       DOI = {10.4310/atmp.251120034604},
       URL = {https://doi.org/10.4310/atmp.251120034604},
    EPRINT = {2407.07131},
}

@misc{mihalache2025constructing,
  title={Constructing solutions of simplex equations from polygon equations},
  author={Mihalache, Serban Matei and Mochida, Tomoro},
  note={To appear in \textit{SIGMA Symmetry Integrability Geom. Methods Appl.} arXiv:2510.12905},
  year={2025}
}

@incollection{dimakis2012kp,
  title={K{P} solitons, higher {B}ruhat and {T}amari orders},
  author={Dimakis, Aristophanes and M{\"u}ller-Hoissen, Folkert},
  booktitle={Associahedra, {T}amari {L}attices and {R}elated {S}tructures: {T}amari {M}emorial {F}estschrift},
  series={Progr. Math.},
  volume={299},
  pages={391--423},
  year={2012},
  publisher={Birkh{\"a}user/Springer, Basel}
}

@misc{poguntke2017higher,
  title={Higher {S}egal structures in algebraic {$K$}-theory},
  author={Poguntke, Thomas},
  note={preprint. arXiv:1709.06510},
  year={2017}
}

@book{dyckerhoff2019higher,
 author = {Dyckerhoff, Tobias and Kapranov, Mikhail},
 title = {Higher {S}egal spaces},
 fseries = {Lecture Notes in Mathematics},
 series = {Lect. Notes Math.},
 issn = {0075-8434},
 volume = {2244},
 isbn = {978-3-030-27122-0; 978-3-030-27124-4},
 year = {2019},
 publisher = {Cham: Springer},
 language = {English},
 doi = {10.1007/978-3-030-27124-4},
 zbMATH = {7103772},
 Zbl = {1459.18001}
}

@incollection{dyckerhoff2026cyclic,
  title={Cyclic polytopes, orientals, and correspondences: some aspects of higher {S}egal spaces},
  booktitle={Higher {S}egal spaces and applications},
  author = {Dyckerhoff, Tobias},
  series = {Contemp. Math.},
  volume = {838},
  pages={131--158},
  publisher = {Amer. Math. Soc., Providence, RI},
  year={2026}
}

@article{danny2012decalage,
  title={Décalage and {K}an's simplicial loop group functor},
  author={Stevenson, Danny},
  journal={Theory Appl. Categ.},
  fjournal={Theory and Applications of Categories},
  volume={26},
  number={28},
  pages={768--787},
  year={2012},
}

@article {williams2023first,
    AUTHOR = {Williams, Nicholas J.},
     TITLE = {The first higher {S}tasheff-{T}amari orders are quotients of
              the higher {B}ruhat orders},
   JOURNAL = {Electron. J. Combin.},
  FJOURNAL = {Electronic Journal of Combinatorics},
    VOLUME = {30},
      YEAR = {2023},
    NUMBER = {1},
     PAGES = {Paper No. 1.29, 38},
      ISSN = {1077-8926},
   MRCLASS = {06A07 (05B45 35C08)},
  MRNUMBER = {4546619},
MRREVIEWER = {Myrto\ Kallipoliti},
       DOI = {10.37236/10877},
       URL = {https://doi.org/10.37236/10877},
}

@article{carter2003quandle,
  title={Quandle cohomology and state-sum invariants of knotted curves and surfaces},
  author={Carter, J and Jelsovsky, Daniel and Kamada, Seiichi and Langford, Laurel and Saito, Masahico},
  JOURNAL = {Trans. Amer. Math. Soc.},
  FJOURNAL = {Transactions of the American Mathematical Society},
  volume={355},
  number={10},
  pages={3947--3989},
  year={2003}
}

@article {carter2004homology,
    AUTHOR = {Carter, J. Scott and Elhamdadi, Mohamed and Saito, Masahico},
     TITLE = {Homology theory for the set-theoretic {Y}ang-{B}axter equation
              and knot invariants from generalizations of quandles},
   JOURNAL = {Fund. Math.},
  FJOURNAL = {Fundamenta Mathematicae},
    VOLUME = {184},
      YEAR = {2004},
     PAGES = {31--54},
      ISSN = {0016-2736,1730-6329},
   MRCLASS = {57M25 (55N35)},
  MRNUMBER = {2128041},
MRREVIEWER = {Seiichi\ Kamada},
       DOI = {10.4064/fm184-0-3},
       URL = {https://doi.org/10.4064/fm184-0-3},
}

@article{colazzo2020set,
  title={Set-theoretic solutions of the pentagon equation},
  author={Colazzo, Ilaria and Jespers, Eric and Kubat, {\L}ukasz},
  journal={Comm. Math. Phys.},
  fjournal={Communications in Mathematical Physics},
  volume={380},
  number={2},
  pages={1003--1024},
  year={2020},
  publisher={Springer}
}

@article{colazzo2024bijective,
  title={Bijective solutions to the pentagon equation},
  author={Colazzo, I and Okni{\'n}ski, J and Van Antwerpen, A},
  journal={arXiv preprint arXiv:2405.20406},
  year={2024}
}

@article {castelli2026commutative,
    AUTHOR = {Castelli, Marco},
     TITLE = {On commutative set-theoretic solutions of the pentagon
              equation},
   JOURNAL = {Semigroup Forum},
  FJOURNAL = {Semigroup Forum},
    VOLUME = {112},
      YEAR = {2026},
    NUMBER = {2},
     PAGES = {379--396},
      ISSN = {0037-1912,1432-2137},
   MRCLASS = {20M05 (20B05 81R05)},
  MRNUMBER = {5019535},
       DOI = {10.1007/s00233-026-10609-7},
       URL = {https://doi.org/10.1007/s00233-026-10609-7},
}

@article{kapranov1991combinatorial,
  author  = {Kapranov, M. M. and Voevodsky, V. A.},
  title   = {Combinatorial-geometric aspects of polycategory theory:
             pasting schemes and higher {B}ruhat orders (list of results)},
  JOURNAL = {Cahiers Topologie G{\'e}om. Diff{\'e}rentielle Cat{\'e}g.},
  fjournal = {Cahiers de Topologie et G{\'e}om{\'e}trie Diff{\'e}rentielle Cat{\'e}goriques},
  volume  = {32},
  number  = {1},
  pages   = {11--27},
  year    = {1991},
}

@article{mazzotta2025set,
  title={Set-theoretical solutions to the pentagon equation: a survey},
  author={Mazzotta, Marzia},
  journal={Commun. Math.},
  fjournal={Communications in Mathematics},
  volume={33},
  year={2025},
  number = {3},
  pages={Paper No. 2, 17},
  publisher={Episciences. org}
}

@article{dijkgraaf1990topological,
    author = {Dijkgraaf, Robbert and Witten, Edward},
    title = {Topological gauge theories and group cohomology},
    doi = {10.1007/BF02096988},
    journal={Comm. Math. Phys.},
    fjournal = {Commun. Math. Phys.},
    volume = {129},
    pages = {393},
    year = {1990}
}

@article{wakui1992dijkgraaf,
  title={On {D}ijkgraaf--{W}itten invariant for 3-manifolds},
  author={Wakui, Michihisa},
  journal={Osaka J. Math.},
  fjournal={Osaka Journal of Mathematics},
  volume={29},
  number={4},
  pages={675--696},
  year={1992},
  publisher={Osaka University and Osaka City University}
}

@misc{korepanov2019polynomial,
  title={Polynomial-valued constant hexagon cohomology},
  author={Korepanov, Igor G},
  note={preprint. arXiv:1904.07000},
  year={2019}
}

@article{korepanov2021nonconstant,
  title={Nonconstant hexagon relations and their cohomology},
  author={Korepanov, Igor G},
  journal={Lett. Math. Phys.},
  fjournal={Letters in Mathematical Physics},
  volume={111},
  number={1},
  pages={Paper No. 1, 24},
  year={2021},
  publisher={Springer}
}

@incollection{gale1963neighborly,
  title={Neighborly and cyclic polytopes},
  author={Gale, David},
  booktitle={Proc. {S}ympos. {P}ure {M}ath., {V}ol. {VII}},
  publisher={Amer. Math. Soc., Providence, RI},
  pages={225--232},
  year={1963}
}

@article{kassotakis2024entwining,
  title={Entwining tetrahedron maps},
  author={Kassotakis, Pavlos},
  fjournal={Partial Differential Equations in Applied Mathematics},
  journal={Partial Differ. Equ. Appl. Math.},
  volume={12},
  pages={100949},
  year={2024},
  publisher={Elsevier}
}

@article{galvez2018decomposition,
  title={Decomposition spaces, incidence algebras and {M}{\"o}bius inversion {I}: basic theory},
  author={G{\'a}lvez-Carrillo, Imma and Kock, Joachim and Tonks, Andrew},
  journal={Adv. Math.},
  fjournal={Advances in Mathematics},
  volume={331},
  pages={952--1015},
  year={2018},
  publisher={Elsevier}
}
\bibliographystyle{plain}

\end{document}